\documentclass[12pt,draftclsnofoot,onecolumn]{IEEEtran}

\usepackage{graphics}
\usepackage{bm}
\usepackage{cite}
\usepackage{subcaption}
\usepackage{amsmath,amsthm, amssymb, latexsym,tikz}
\newcommand{\e}{\mathrm{e}}
\newcommand{\eqdef}{\stackrel{\textrm{\tiny def}}{=}}

\newcommand*{\point}{\makebox[1ex]{\textbf{$\cdot$}}}%

\usepackage[font=small]{caption}

\usepackage{array,balance}
\usepackage[flushleft]{threeparttable}

\ifCLASSINFOpdf
\else
\fi
\usepackage{comment}
\usepackage{amsthm}

\newtheorem{thm}{Theorem}

\begin{document}
\title{Downlink Single-Snapshot Localization and Mapping with a Single-Antenna Receiver}

\author{\IEEEauthorblockN{Alessio Fascista, \IEEEmembership{Member, IEEE}, Angelo Coluccia, \IEEEmembership{Senior Member, IEEE}, Henk Wymeersch, \IEEEmembership{Senior Member, IEEE}, and Gonzalo Seco-Granados, \IEEEmembership{Senior Member, IEEE}} % <-this % stops a space
\thanks{This work was supported in part by the Swedish Research Council (VR) under project No.~2018-03701.}% <-this % stops a space
\thanks{A. Fascista and A. Coluccia are with the Department of Innovation Engineering, Universit\`a del Salento, Via Monteroni, 73100 Lecce, Italy (e-mail: alessio.fascista@unisalento.it; angelo.coluccia@unisalento.it).}
\thanks{H. Wymeersch is with the Department of Electrical Engineering, Chalmers University of Technology, 412 96 Gothenburg, Sweden (e-mail: henkw@chalmers.se).}
\thanks{G. Seco-Granados is with the Department of Telecommunications and Systems Engineering, Universitat Aut\`onoma de Barcelona, 08193 Barcelona, Spain (e-mail: gonzalo.seco@uab.cat).}
}

\maketitle

\begin{abstract}
5G mmWave MIMO systems enable accurate estimation of the user position and mapping of the radio environment using a single snapshot when both the base station (BS) and user are equipped with large antenna arrays. However, massive arrays are initially expected only at the BS side, likely leaving users with one or very few antennas. In this paper, we propose a novel method for single-snapshot localization and mapping in the more challenging case of a user equipped with a single-antenna receiver. The joint maximum likelihood (ML) estimation problem is formulated and its solution formally derived. To avoid the burden of a full-dimensional search over the space of the unknown parameters, we present a novel practical approach that exploits the sparsity of mmWave channels to compute an approximate joint ML estimate. A thorough analysis, including the derivation of the Cram\'er-Rao lower bounds, reveals that  accurate localization and mapping can be achieved also in a MISO setup even when the direct line-of-sight path between the BS and the user is severely attenuated.
\end{abstract}


% \fi
\IEEEpeerreviewmaketitle

\section{Introduction}
\IEEEPARstart{T}{he} advent of fifth-generation (5G) mobile cellular communications is paving the way for a technological revolution \cite{rappaport_5G,Witrisal}. Millimeter wave (mmWave) signals and massive multiple-input multiple-output (MIMO) technologies are regarded as key pillars of emerging 5G systems, thanks to the expected high data rates and spectral efficiency \cite{Swindle,Heath,Tufve}. Large bandwidths and massive antenna arrays make also possible very precise estimation of location-related information such as time-of-flight (TOF), angle-of-arrival (AOA), and angle-of-departure (AOD), which can enable applications requiring accurate localization \cite{5G_vehic,Peral,Kakkavas}. 

The localization capabilities of mmWave MIMO systems have received significant attention. In \cite{AbuShaban} the Cram\'{e}r-Rao Lower bound (CRLB) for the problem of 3D localization is derived, highlighting the main differences in achievable accuracy between uplink (UL) and downlink (DL) channels. The theoretical analysis revealed that mmWave MIMO systems can provide cm-level accuracy even when the positioning process is supported by a single base station (BS). Over the last years, a number of localization algorithms have appeared in the literature \cite{Arash,Nil,ESPRIT,Koivisto}.

Differently from conventional radio-frequency systems, the peculiar characteristics of mmWave MIMO channels make it possible to estimate position-related parameters for each received non-line-of-sight (NLOS) path \cite{Palacios}; remarkably, the Fisher information analysis in \cite{Mendrzik} revealed that NLOS components provide additional information over the line-of-sight (LOS) path, which can be fruitfully leveraged to improve the localization performance. 
%indeed, thanks to the high temporal and spatial resolution, the TOFs, AOAs and AODs originating from multipath propagation can be linked to the positions of BSs, users, and physical scatterers or reflectors in the surrounding environment, through a simple geometrical model. 
In addition to accurately localizing one or more users, mmWave MIMO can be also exploited  to progressively build a map of the radio environment over time, a problem that can be categorized as a simultaneous localization and mapping (SLAM) problem (for more details on the topic, please refer to \cite{SLAM1,SLAM2}). 
%Such maps can be used, for instance, to further refine the position estimation as well as to solve for the synchronization of users \cite{Henk5G}.
%especially if such an information is shared cooperatively in the network.
A few papers have recently started to address this problem, specifically to exploit NLOS paths for both position estimation and  mapping in mmWave MIMO \cite{Destino,Mendrzik2,Battistelli,Leitinger2}. Thanks to the high temporal and spatial resolution, the TOFs, AOAs and AODs originating from multipath propagation can be directly linked to the positions of BSs, users, and physical scatterers or reflectors at each time instant, allowing the SLAM problem to be solved using only a single snapshot of the environment. 

While mmWave MIMO enables high positioning and mapping  accuracy with a single snapshot, 
it  requires the deployment of large-scale antenna arrays at the user side, considerably increasing the complexity and cost of the overall system. In contrast, 
 mobile users using smartphones, as well as wearable/portable IoT devices, will be initially equipped with one or very few antennas \cite{Westberg}. 
Localization and mapping using a single antenna at both transmit and receive side, namely single-input single-output (SISO), has been addressed in the context of ultrawide-band systems. Differently from the MIMO case, only TOF or RSS information can be used in the estimation process, which in turn requires multiple snapshots corresponding to different positions of the user to lead to an identifiable SLAM solution 
%must be collected and processed over time, typically within a Bayesian tracking framework 
\cite{Gentner,Leitinger_SLAM}.

\begin{comment}
The main common characteristic of the above  methods resides in the need of  massive antenna arrays at both BS and user side in order to achieve high positioning and mapping  accuracy. 
Although a (possibly massive) MIMO setup enables highly bidirectional transmissions to compensate for the severe path losses at mmWave frequencies \cite{Sayeed}, it also requires the deployment of large-scale antenna arrays at the user side, considerably increasing the complexity and cost of the overall system.  Furthermore, it is expected that in the very near future, massive arrays will be initially available only at the BS side, while users (especially mobile users using smartphones, as well as wearable/portable IoT devices) will be likely equipped with one or very few antennas \cite{Westberg}. 

Localization and mapping using a single antenna at both transmit and receive side, namely single-input single-output (SISO), has been widely addressed in the context of ultrawide-band systems. Differently from MIMO SLAM, which can be performed using a single snapshot of the environment, the lack of antenna arrays in SISO systems implies that only TOF or RSS information can be used in the estimation process. Consequently, to make the localization and mapping problem solvable, multiple snapshots corresponding to different positions of the user must be collected and processed over time, typically within a Bayesian tracking framework \cite{Gentner,Leitinger_SLAM}. 
\end{comment}

In this paper, we aim to partially close the knowledge gap between MIMO and SISO systems, and investigate the problem of  single-snapshot localization and mapping in the challenging multiple-input single-output (MISO) case of a mmWave system, where the user is equipped with a single-antenna receiver while the BS has a transmit array. Specifically, we exploit the TOF and AOD information associated to the DL signals transmitted from a single BS, allowing single-snapshot localization and mapping, even in the presence of NLOS paths.  The use of DL as opposed to UL signals leads to better signal-to-noise ratio (SNR) conditions for the estimation problem \cite{ICASSP2020}. 
The main contributions of this work are as follows:
\begin{itemize}
    \item \emph{A fundamental Fisher information analysis} is conducted, which allows to understand the problem from a theoretical perspective, extending the CRLB analysis for the LOS-only scenario in  \cite{TWC_2019} with a thorough evaluation of the achievable performance when the NLOS paths are explicitly taken into account in the estimation process. Remarkably, we will show that accurate single-snapshot localization and mapping is still possible in a MISO setup, but in contrast to the MIMO case, map information does not increase the user position information;
    \item \emph{The derivation of the maximum likelihood (ML) estimator for localization and mapping} is provided, showing the equivalence of channel-domain and position-domain formulations. Furthermore, we show that mapping of the scatterers positions  depends also on the estimation accuracy of LOS parameters, in line with Fisher information analysis. 
   \item \emph{A low-complexity estimator} is proposed, by exploiting the sparsity of the mmWave channel.  We propose an efficient two-step algorithm which allows the computation of an accurate approximate solution of the joint ML estimation problem, but avoiding the need of a full-dimensional search in the space of the unknown parameters.
\end{itemize}

A thorough simulation analysis demonstrates that the proposed joint ML algorithm enables a very accurate estimation of the user position and mapping of the scatterers locations, with performances attaining the theoretical lower bounds even when the LOS path is severely attenuated.

The rest of the paper is organized as follows. In Sec.~II, we introduce the system model and describe in details the reference scenario. In Sec.~III, we derive and analyze the fundamental bounds on the estimation of the channel and location parameters in the considered MISO setup. Then, in Sec.~IV we formulate the joint ML estimation problem in the channel domain and propose a novel low-complexity localization and mapping approach; furthermore, we discuss the equivalence with the joint ML estimator in the position domain. The performance of the proposed approach is then assessed in Sec.~V. We conclude the paper in Sec.~VI.

\section{System Model}
The reference scenario addressed in this paper consists of a MISO system in which a BS, equipped with $N_{\text{\tiny BS}}$ antennas, communicates with a mobile station (MS) equipped with a single antenna receiver. The system operates at a carrier frequency $f_c$ (corresponding to wavelength $\lambda_c$) and uses signals having bandwidth $B$. Without loss of generality, the BS is located in the origin, i.e., $\bm{p}_\text{\tiny BS} = [0 \ 0]^{\text T}$, while we denote by $\bm{p} = [p_x \ p_y]^{\text T}$ the unknown position of the MS.

\subsection{Transmitter Model}
We consider the transmission of orthogonal frequency division multiplexing (OFDM) signals. Particularly, we assume that $G$ signals are broadcast in DL sequentially, with the $g$-th transmission consisting in $M$ simultaneously transmitted symbols over each subcarrier $n=0,\ldots,N-1$, i.e.,
%\begin{equation}
$\bm{x}^{g}[n] = \left[x_1[n]\ \cdots\ x_M[n]\right]^{\text T} \in \mathbb{C}^{M \times 1}$, 
%\end{equation}
with $P_t = \mathbb{E}\left[\|\bm{x}^{g}[n]\|^2\right]$ the transmitted power and $\mathbb{E}[\point]$ denoting the expectation operator. After precoding, the symbols are transformed to the time-domain using an $N$-point Inverse Fast Fourier Transform (IFFT). A cyclic prefix (CP) of length $T_{\text{CP}} = DT_S$ is added before the radio-frequency (RF) precoding, with $D$ number of symbols in the CP and $T_S = 1/B$  the sampling period. %Hereafter, we assume that $T_{\text{CP}}$ exceeds the delay spread of the channel. 

The signal transmitted over subcarrier $n$ at time $g$ is expressed as $\bm{z}^{g}[n]=\bm{F}^{g}[n]\bm{x}^{g}[n]$, with $\bm{F}^g[n] \in \mathbb{C}^{N_{\text{\tiny BS}} \times M}$ denoting the beamforming matrix applied at the transmit side. In absence of a priori knowledge about the user location, the $M$ beams in the beamforming matrix are typically set to ensure a uniform coverage of the considered area.
Furthermore, a total power constraint  $\|\bm{F}^g[n]\|_{\text{F}} = 1$ is imposed to the transmit beamforming \cite{Alkhateeb1}. Given the typical sparsity of the mmWave channels, less beams than antenna elements can be considered, i.e., $M \leq N_{\text{\tiny BS}}$ \cite{brady1,brady2}. 

\subsection{Channel Model}
We assume that a direct LOS link exists between the BS and the MS, and that additional NLOS paths due to local scatterers or reflectors may also be present.
\begin{figure}
\centering
 \includegraphics[width=0.5\linewidth]{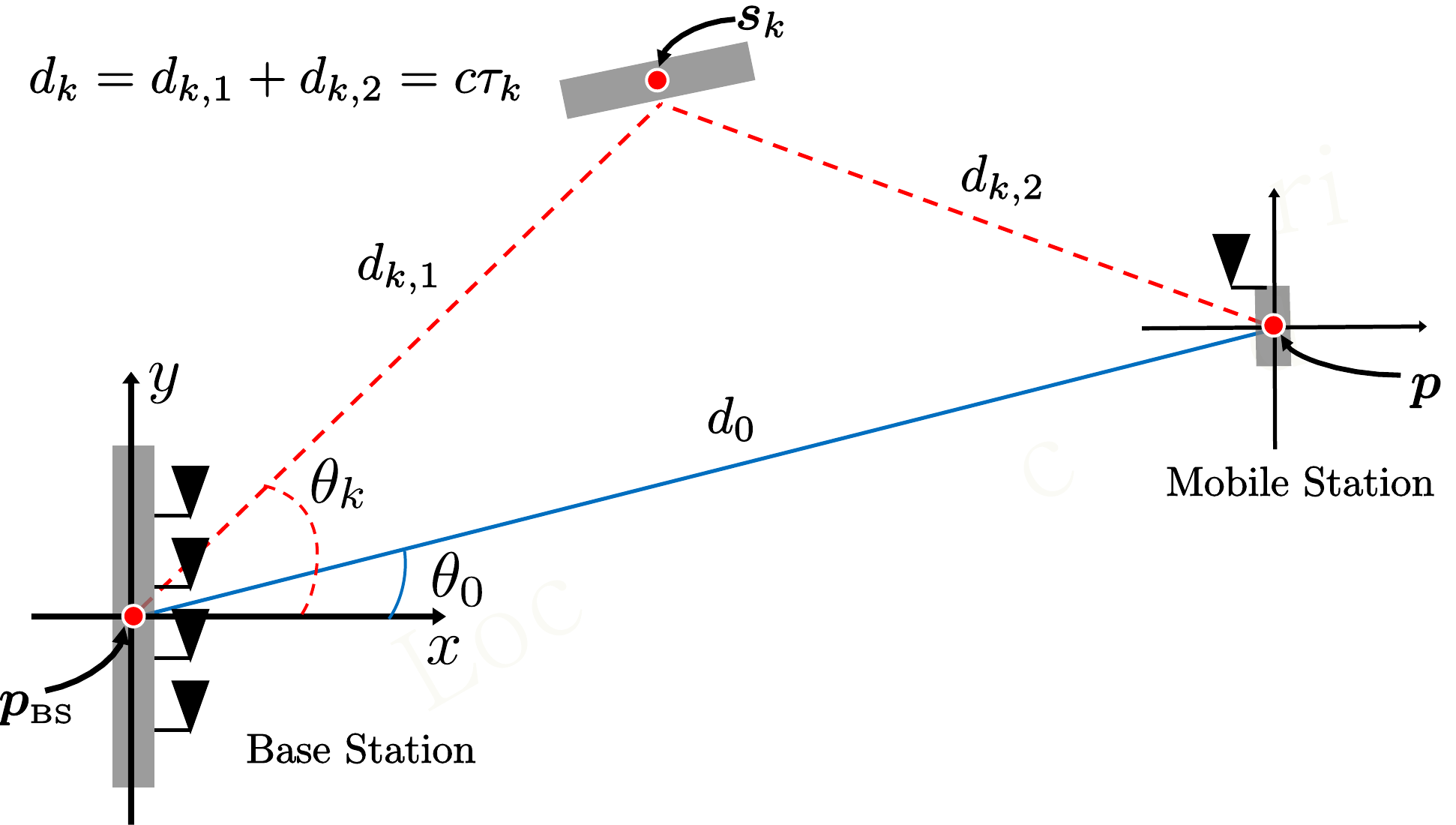}
 	\caption{Geometry of the considered two-dimensional localization and mapping scenario.}
\label{fig:scenario}
 \end{figure}
 For the sake of the analysis, we also assume that the system has been %preliminarily 
 synchronized during an initial phase using, e.g., a two-way protocol \cite{AbuSync,Sark}.
The different position-related parameters of the channel are depicted in Fig.~\ref{fig:scenario}. These parameters include $\theta_{k}$ and $\tau_k$, denoting the AOD and TOF related to the $k$-th path, respectively. In the following, $k = 0$ corresponds to the LOS link and $k \geq 1$ denotes the NLOS paths. Moreover, we denote by $\bm{s}_k = [s_{k,x} \ s_{k,y}]^{\text T}$ the unknown position of the scatterer giving rise to $k$-th NLOS path, for which $d_{k,1} = \| \bm{s}_k - \bm{p}_\text{\tiny BS}\| = \| \bm{s}_k \|$ and $d_{k,2} = \|\bm{p} - \bm{s}_k\| $, with $\| \point \|$ denoting the Euclidean distance. We consider  by convention $\bm{s}_0 \equiv \bm{p}$, making all expressions well-defined also for $k=0$. 
Assuming $K+1$ paths, the $1 \times N_{\text{\tiny BS}}$ complex channel vector associated with subcarrier $n$ is given by
\begin{equation}\label{eq::fullchannelmodel}
\bm{h}^{\text T}[n] = \bm{\zeta}^{\text T}[n]\bm{A}^{\text H}_{\text{\tiny BS}}
\end{equation}
where we leveraged $\lambda_n = c/(\frac{n}{NT_S} + f_c) \approx \lambda_c \,\forall n$ (with $c$ denoting the speed of light), i.e., the typical narrowband condition. The array response matrix is given by
\begin{align}
    \bm{A}_{\text{\tiny BS}} & = [\bm{a}_{\text{\tiny BS}}(\theta_0), \ldots,\bm{a}_{\text{\tiny BS}}(\theta_K) ]
    %\bm{\zeta}[n] & = \sqrt{N_{\text{\tiny BS}}} \times \begin{bmatrix}
    %\alpha_0 \e^{\frac{-j2\pi n \tau_0}{N T_S}} \\
     %\vdots \\
    %\alpha_K \e^{\frac{-j2\pi n \tau_K}{N T_S}} 
  %\end{bmatrix},
\end{align}
and $[\bm{\zeta}[n]]_k=\sqrt{N_{\text{\tiny BS}}}\alpha_k \e^{\frac{-j2\pi n \tau_k}{N T_S}}$, 
%\begin{equation}
%\bm{A}_{\text{\tiny BS}} = [\bm{a}_{\text{\tiny BS}}(\theta_0), \ldots,\bm{a}_{\text{\tiny BS}}(\theta_K) ]
%\end{equation}
% \begin{equation}
% \bm{A}_{\text{Rx}}[n] = [\bm{a}_{\text{Rx},n}(\vartheta_0), \ldots,\bm{a}_{\text{Rx},n}(\vartheta_P) ]
% \end{equation}
%and
%\begin{equation}
%\bm{\Gamma}[n] = \sqrt{N_{\text{\tiny BS}}} \times \begin{bmatrix}
 %   \alpha_0 \e^{\frac{-j2\pi n \tau_0}{N T_S}} \\
  %   \vdots \\
   % \alpha_K \e^{\frac{-j2\pi n \tau_K}{N T_S}} 
  %\end{bmatrix}
%\end{equation}
% \begin{equation}
% \bm{h}^T[n] = \sqrt{M}\alpha \e^{\frac{-j2\pi n \tau}{NT_S}}\bm{a}^{H}_{\text{Tx},n}(\theta)
% \end{equation}
where $\alpha_k = h_k/\sqrt\rho_k$, with $\rho_k$ the path loss and $h_k$ denoting the complex channel gain of the $k$-th path, respectively.
%The structure of the antenna steering vectors $\bm{a}_{\text{\tiny BS}}(\theta_k) \in \mathbb{C}^{N_{\text{\tiny BS}} \times 1}$ depends on the specific geometry of the considered array. 
Without loss of generality, in the following we consider a Uniform Linear Array (ULA) without mutual antenna coupling and with isotropic antennas, whose steering vector can be expressed as
%\footnote{The response vector  is defined in a similar way, i.e., $\bm{a}_{\text{Rx},n}(\vartheta) = \frac{1}{\sqrt{M_r}} \left[1 \ e^{j\frac{2\pi}{\lambda_n} d\sin\vartheta} \cdots \ e^{j(M_r-1)\frac{2\pi}{\lambda_n}d\sin\vartheta}\right]^T$.} 
\begin{equation}\label{eq::steervect}
\bm{a}_{\text{\tiny BS}}(\theta) = \frac{1}{\sqrt{N_{\text{\tiny BS}}}} \left[1 \ e^{j\frac{2\pi}{\lambda_c} d\sin\theta} \cdots \ e^{j(N_{\text{\tiny BS}}-1)\frac{2\pi}{\lambda_c}d\sin\theta}\right]^{\text T}
\end{equation}
%where $\lambda_n = c/\left(\frac{n}{NT_S} + f_c\right)$ is the signal wavelength at the $n$-th subcarrier, $c$ is the speed of light,
where $d = \frac{\lambda_c}{2}$ denotes the ULA interelement spacing. 
%We observe that for $B \ll f_c$, $\lambda_n \approx \lambda_c \forall n$, and \eqref{eq::steervect} reverts to the well-known narrowband model.

\subsection{Received Signal Model}
The received signal related to the $n$-th subcarrier and transmission $g$, after CP removal and Fast Fourier Transform (FFT), is given by
\begin{equation}\label{eq::recsignals}
%\bm{y}^{g}[n] = \bm{H}[n]\bm{F}^{g}[n]\bm{x}^{g}[n] + \bm{\nu}^{g}[n]
y^{g}[n] = \bm{h}^{\text T}[n]\bm{F}^{g}[n]\bm{x}^{g}[n] + \nu^{g}[n]
\end{equation}
where $\nu^g[n]$ is the additive circularly complex Gaussian noise with zero mean and variance $\sigma^2$. The objective of the paper is to determine the unknown MS position $\bm{p}$ as well as to map the location of the scatters $\bm{s}_k$, $k \geq 1$ present in the environment from the set of all received signals 
\begin{equation}\label{eq::Y}
\bm{Y} = \begin{bmatrix}
    y^1[0] & \cdots & y^G[0] \\
    \vdots& \ddots & \vdots \\
    y^1[N-1] & \cdots & y^G[N-1] 
  \end{bmatrix}.
\end{equation}

\section{MISO: Fundamental Bounds in Multipath Scenario}\label{sec:bounds}

In this section, we derive the expressions of the Fisher Information Matrix (FIM) and CRLB related to the estimation of the MS position $\bm{p}$ and scatterers positions $\bm{s}_k$. As a first step, we evaluate the theoretical bounds on the estimation of the channel parameters (i.e., AODs, TOFs, and channel gains). Subsequently, such bounds are transformed in the position domain and further analyzed to gain insights on the achievable performance in terms of joint localization of the user and mapping of the environment.

\subsection{FIM on Channel Parameters}
Let $\bm{\gamma} \in \mathbb{R}^{4(K+1) \times 1}$ denotes the vector of the unknown channel parameters
%\begin{equation}
$\bm{\gamma} = [\bm{\gamma}_0^{\text T} \cdots \bm{\gamma}_K^{\text T}]^{\text T}$, 
%\end{equation}
where each $\bm{\gamma}_k$ consists of the channel complex amplitude, TOF and AOD for the $k$-th path and is given by
%\begin{equation}
$\bm{\gamma}_k = [r_k \ \phi_k \ \tau_k \ \theta_k]^{\text T}$.
%\end{equation}
Defining $\hat{\bm{\gamma}}$ as an unbiased estimator of $\bm{\gamma}$, it is well-known that the mean squared error (MSE) is lower bounded as 
\begin{equation}
\mathbb{E}_{\bm{Y}|\bm{\gamma}}\left[(\hat{\bm{\gamma}} - \bm{\gamma})(\hat{\bm{\gamma}} - \bm{\gamma})^{\text T} \right] \succeq \bm{J}^{-1}_{\bm{\gamma}}
\end{equation}
where $\mathbb{E}_{\bm{Y}|\bm{\gamma}}[\point]$ denotes the expectation parameterized as function of the unknown vector $\bm{\gamma}$ and $\bm{J}_{\bm{\gamma}}$ is the $4(K+1) \times 4(K+1)$ FIM defined as $\bm{J}_{\bm{\gamma}} = \mathbb{E}_{\bm{Y}|\bm{\gamma}}\left[-\frac{\partial^2 \log f(\bm{Y}|\bm{\gamma})}{\partial \bm{\gamma}\partial \bm{\gamma}^{\text T}} \right].$
%\begin{equation}
%\bm{J}_{\bm{\gamma}} = \mathbb{E}_{\bm{Y}|\bm{\gamma}}\left[-\frac{\partial^2 \log f(\bm{Y}|\bm{\gamma})}{\partial \bm{\gamma}\partial \bm{\gamma}^{\text T}} \right] \label{eq::FIM}.
%\end{equation}
The FIM  can be structured as
\begin{equation}\label{eq::FIMchannel}
\bm{J}_{\bm{\gamma}} = \left[\begin{matrix}
\bm{\Lambda}(\bm{\gamma}_0,\bm{\gamma}_0) & \cdots & \bm{\Lambda}(\bm{\gamma}_0,\bm{\gamma}_K) \\
\vdots & \ddots & \vdots \\
\bm{\Lambda}(\bm{\gamma}_K,\bm{\gamma}_0) & \cdots & \bm{\Lambda}(\bm{\gamma}_K,\bm{\gamma}_K) \\
\end{matrix}
\right]
\end{equation}
where the $4 \times 4$ matrix $\bm{\Lambda}(\bm{\gamma}_h,\bm{\gamma}_\ell)$ is given by
\begin{align}\label{eq::FIMblock}
\bm{\Lambda}(\bm{\gamma}_h,\bm{\gamma}_\ell) &= \mathbb{E}_{\bm{Y}|\bm{\gamma}}\left[-\frac{\partial^2 \log f(\bm{Y}|\bm{\gamma})}{\partial \bm{\gamma}_h\partial \bm{\gamma}_\ell^{\text T}} \right] \nonumber \\
& \!\!\!\!\!\!\!\!\!\!\!\!\!\!\!\!\!\!\!\!= \left[\begin{matrix}
\Lambda(r_h,r_\ell) & \Lambda(r_h,\phi_\ell) & \Lambda(r_h,\tau_\ell) & \Lambda(r_h,\theta_\ell) \\
\Lambda(\phi_h,r_\ell) & \Lambda(\phi_h,\phi_\ell) & \Lambda(\phi_h,\tau_\ell) & \Lambda(\phi_h,\theta_\ell) \\
\Lambda(\tau_h,r_\ell) & \Lambda(\tau_h,\phi_\ell) & \Lambda(\tau_h,\tau_\ell) & \Lambda(\tau_h,\theta_\ell) \\
\Lambda(\theta_h,r_\ell) & \Lambda(\theta_h,\phi_\ell) & \Lambda(\theta_h,\tau_\ell) & \Lambda(\theta_h,\theta_\ell)
\end{matrix} \right]
\end{align}
with $h,\ell = 0,\ldots,K$. Substituting $\bm{Y}$ from \eqref{eq::Y} in \eqref{eq::FIMblock} and accounting for the noise statistics yields %and considering the fact that $\mathbb{E}[\nu^g[n]] = 0$ $\forall g=1,\ldots,G$, and $\forall n=0,\ldots,N-1$, it follows that each generic entry $\Lambda(x_h,x_\ell)$ is expressed as
\begin{equation}\label{eq::FIMelemformula}
\Lambda(\bm{\gamma}_h,\bm{\gamma}_\ell) = \frac{2}{\sigma^2} \sum_{g=1}^G \sum_{n=0}^{N-1} \Re\left\{\left(\frac{\partial m^g[n]}{\partial \bm{\gamma}_h}\right)^{\text H} \frac{\partial m^g[n]}{\partial \bm{\gamma}_\ell}\right\}
\end{equation}
with $\Re\{\point\}$ denoting the real-part operator and the noise-free observation at subcarrier $n$, transmission $g$
is 
$m^{g}[n]=\sqrt{N_{\text{\tiny BS}}} \sum_{k=0}^K\alpha_k\exp\left(\frac{-j2 \pi n\tau_k}{NT_S}\right)\bm{a}^{\text H}_{\text{\tiny BS}}(\theta_k)\bm{z}^{g}[n]$, 
with $\alpha_k= h_k/\sqrt\rho_k \eqdef r_ke^{j\phi_k}$ with $r_k$ and $\phi_k$ modulus and phase of the complex amplitude $\alpha_k$, respectively. 
We report in Appendix \ref{appA} the value of each entry of $\bm{\Lambda}(\bm{\gamma}_h,\bm{\gamma}_\ell)$. 
Two paths $h$ and $\ell$ are said to be orthogonal when $\bm{\Lambda}(\bm{\gamma}_h,\bm{\gamma}_\ell)=\bm{0}_{4 \times 4}$ \cite{cox1987parameter}, with $\bm{0}_{L \times L}$ a $L \times L$ matrix of zeros.

\subsection{FIM on Position Parameters}
In this section, we derive the FIM in the position domain by applying a transformation of variables from the vector of channel parameters  $\bm{\gamma}$ to a new vector of location parameters $\bm{\eta} = [\bm{\eta}_0^{\text T} \cdots \bm{\eta}_K^{\text T}]^{\text T}$, where
%\begin{equation}
$\bm{\eta}_0 =  [r_0 \ \phi_0 \ p_x \ p_y]^{\text T}$  and  $\bm{\eta}_k = [r_k \ \phi_k \ s_{k,x} \ s_{k,y}]^{\text T}$, for $k \geq 1$.
%\end{equation}
More specifically, by exploiting the geometric relationships between the parameters in $\bm{\gamma}$ and $\bm{\eta}$, we have 
\begin{align}
    \tau_0 & = \| \bm{p}\|/c \label{eq:tau0}\\
    \theta_0 & = \mathrm{atan2}(p_y,p_x) \label{eq:theta0}\\
    \tau_k & = \|\bm{s}_k\|/c + \|\bm{p} - \bm{s}_k \|/c, \;\; k\geq 1  \label{eq:tauk}\\
\theta_k & = \mathrm{atan2}(s_{k,y},s_{k,x}), \;\; k\geq 1, \label{eq:thetak}
\end{align}
%
%\begin{equation}\label{eq:tau0}
%\tau_0 = \| \bm{p}\|/c
%\end{equation}
%\begin{equation}\label{eq:theta0}
%\theta_0 = \mathrm{atan2}(p_y,p_x)  
%\end{equation}
%\begin{equation}\label{eq:tauk}
%\tau_k = \|\bm{s}_k\|/c + \|\bm{p} - %\bm{s}_k \|/c, \;\; k\geq 1  
%\end{equation}
%\begin{equation}\label{eq:thetak}
%\theta_k = %\mathrm{atan2}(s_{k,y},s_{k,x}), \;\; k\geq 1 
%\end{equation}
 where the function $\mathrm{atan2}(y,x)$ is the four-quadrant inverse tangent, and the angles are measured counterclockwise with respect to the $x$-axis.
 
The FIM in the position space $\bm{\eta}$ is obtained by means of the $4(K+1) \times 4(K+1)$ transformation matrix $\bm{T}$ as
\begin{equation}\label{eq::FIMpos}
\bm{J}_{\bm{\eta}} = \bm{T}\bm{J}_{\bm{\gamma}}\bm{T}^{\text{T}}    
\end{equation}
where
\begin{equation}
   \bm{T} \eqdef \frac{\partial \bm{\gamma}^{\text{T}}}{\partial \bm{\eta}} =  \begin{bmatrix}
   \bm{T}_{0,0} & \ldots & \bm{T}_{K,0}\\
   \vdots & \ddots & \vdots \\
   \bm{T}_{0,K} & \ldots & \bm{T}_{K,K}
   \end{bmatrix}.
\end{equation}
and each submatrix $\bm{T}_{h,\ell}$, $h,\ell=0,\ldots,K$, %for $k \neq 0$ 
is given by
\begin{align}
\bm{T}_{h,\ell}  \eqdef \frac{\partial \bm{\gamma}_h^{\text{T}}}{\partial \bm{\eta}_\ell} &= \begin{bmatrix}
\partial r_h/\partial r_\ell & \partial \phi_h/\partial r_\ell & \partial \tau_h/\partial r_\ell & \partial \theta_h/\partial r_\ell\\
\partial r_h/\partial \phi_\ell & \partial \phi_h/\partial \phi_\ell & \partial \tau_h/\partial \phi_\ell & \partial \theta_h/\partial \phi_\ell\\
\partial r_h/\partial \bm{s}_\ell & \partial \phi_h/\partial \bm{s}_\ell & \partial \tau_h/\partial \bm{s}_\ell & \partial \theta_h/\partial \bm{s}_\ell
\end{bmatrix} \nonumber\\
&  = \begin{bmatrix}
\delta_{h \ell} & 0 & 0 & 0\\
0 & \delta_{h \ell} & 0 & 0\\
0 & 0 & \partial \tau_h/\partial \bm{s}_\ell & \partial \theta_h/\partial \bm{s}_\ell
\end{bmatrix}
\end{align}
where $\delta_{h \ell}$ is the Kronecker symbol and
%while the remaining blocks $\bm{T}_{k,0}$ are obtained as
%\begin{equation}
%\bm{T}_{k,0} = \begin{bmatrix}
%\partial r_k/\partial r_0 & \partial %\phi_k/\partial r_0 & \partial %\tau_k/\partial r_0 & \partial %\theta_k/\partial r_0\\
%\partial r_k/\partial \phi_0 & %\partial \phi_k/\partial \phi_0 & %\partial \tau_k/\partial \phi_0 & %\partial \theta_k/\partial \phi_0\\
%\partial r_k/\partial \bm{p} & %\partial \phi_k/\partial \bm{p} & %\partial \tau_k/\partial \bm{p} & %\partial \theta_k/\partial \bm{p}
%\end{bmatrix}    
%\end{equation}
%where
%\begin{equation}
%\partial r_h/\partial r_\ell = %\partial \phi_h/\partial \phi_\ell %= 1 \;\; \text{for}\; h = \ell
%\end{equation}
\begin{align}
    \frac{\partial \tau_0}{\partial \bm{p}} & =  \frac{1}{c} \left[\frac{p_x}{\|\bm{p}\|} \ \frac{p_y}{\|\bm{p}\|} \right]^{\text T}   \nonumber \\
    \frac{\partial \theta_0}{\partial \bm{p}} & = \left[\frac{-p_y/p^2_x}{1 + (p_y/p_x)^2} \ \frac{1/p_x}{1 + (p_y/p_x)^2} \right]^{\text T}\nonumber\\
    \frac{\partial \tau_h}{\partial \bm{p}}& = \frac{1}{c} \left[\frac{p_x-s_{h,x}}{\|\bm{p}-\bm{s}_h\|} \ \frac{p_y-s_{h,y}}{\|\bm{p}-\bm{s}_h\|} \right]^{\text T} \nonumber\\
    \frac{\partial \tau_h}{\partial \bm{s}_h}& = \frac{1}{c} \left[\left(\frac{s_{h,x}}{\|\bm{s}_h\|} - \frac{(p_x - s_{h,x})}{\|\bm{p} -\bm{s}_h \|} \right)  \left(\frac{s_{h,y}}{\|\bm{s}_k\|} - \frac{(p_y - s_{h,y})}{\|\bm{p} -\bm{s}_h \|} \right) \right]^{\text T}\nonumber\\
    \frac{\partial \theta_h}{\partial \bm{s}_h} & = \left[\frac{-s_{h,y}/s^2_{h,x}}{1 + (s_{h,y}/s_{h,x})^2} \ \frac{1/s_{h,x}}{1 + (s_{h,y}/s_{h,x})^2} \right]^{\text T}, \nonumber 
\end{align}
%
%\begin{equation}
%\partial \tau_0/\partial \bm{p} =  \frac{1}{c} \left[\frac{p_x}{\|\bm{p}\|} \ \frac{p_y}{\|\bm{p}\|} \right]^{\text T}    
%\end{equation}
%\begin{equation}
%\partial \theta_0/\partial \bm{p} = \left[\frac{-p_y/p^2_x}{1 + (p_y/p_x)^2} \ \frac{1/p_x}{1 + (p_y/p_x)^2} \right]^{\text T}   
%\end{equation}
%\begin{equation}
%\partial \tau_h/\partial \bm{p} = \frac{1}{c} \left[\frac{p_x-s_{h,x}}{\|\bm{p}-\bm{s}_h\|} \ \frac{p_y-s_{h,y}}{\|\bm{p}-\bm{s}_h\|} \right]^{\text T} 
%\end{equation}
%\begin{align}
%\partial \tau_h/\partial \bm{s}_h = %\frac{1}{c} &\left[\left(\frac{s_{h,x}}{\|\bm{s}_h\|} - \frac{(p_x - s_{h,x})}{\|\bm{p} -\bm{s}_h \|} \right) \right. \nonumber \\  & \left. \left(\frac{s_{h,y}}{\|\bm{s}_k\|} - \frac{(p_y - s_{h,y})}{\|\bm{p} -\bm{s}_h \|} \right) \right]^{\text T}
%\end{align}
%\begin{equation}
% \partial \theta_h/\partial \bm{s}_h = \left[\frac{-s_{h,y}/s^2_{h,x}}{1 + (s_{h,y}/s_{h,x})^2} \ \frac{1/s_{h,x}}{1 + (s_{h,y}/s_{h,x})^2} \right]^{\text T}  
%\end{equation}
with the last two equations meant for $h\neq 0$, and $\bm{T}_{h,\ell} = \bm{0}_{4 \times 4}$ for $\ell\ge 1$ and $\ell\neq h$.

\subsection{Bounds on MS Position Estimation Error}
To derive the lower bound on the uncertainty of MS position estimation, we consider the CRLB in the location domain obtained by inverting the FIM $\bm{J}_{\bm{\eta}}$ in \eqref{eq::FIMpos}, i.e., $\bm{\Sigma}_p = \bm{J}^{-1}_{\bm{\eta}}$. Specifically, the position error bound (PEB) is computed by adding the third and fourth diagonal entries of the $\bm{\Sigma}_p$ matrix, and taking the square root as
\begin{equation}\label{eq::PEB}
\text{PEB} = \sqrt{\left[\bm{\Sigma}_p\right]_{3,3} + \left[\bm{\Sigma}_p\right]_{4,4}}
\end{equation}
where $[\point]_{j,j}$ selects the $j$-th diagonal entry of $\bm{\Sigma}_p$.

%QUI stava la figure* sul doppiacolonna

\subsection{Role of NLOS Components on MS Position Estimation}\label{sec::roleNLOS}
In the previous subsections, we have derived the fundamental bounds on the estimation of the unknown channel and position parameters. Based on that, we now discuss how the presence of NLOS paths impacts on the estimation of the MS position $\bm{p}$ under the considered MISO setup. We start by recalling that the CRLB matrix $\bm{\Sigma}_p$ matrix is given by $\bm{\Sigma}_p = (\bm{T}\bm{J}_{\bm{\gamma}}\bm{T}^\text{T})^{-1}$. 
%\begin{equation}
%\bm{\Sigma}_p = (\bm{T}\bm{J}_{\bm{\gamma}}\bm{T}^\text{T})^{-1}.
%\end{equation}
Focusing on the vectors $\bm{\gamma}_k$ and $\bm{\eta}_k$ of the $k$-th path, it is interesting to note that the number of parameters in both channel and location domains is the same, and there exists a bijective relationship between them (see eqs.~\eqref{eq:tau0}--\eqref{eq:thetak}). Consequently, $\bm{\Sigma}_p$ can be equivalently expressed as $\bm{\Sigma}_p = (\bm{T}^{-1})^\text{T}\bm{J}^{-1}_{\bm{\gamma}}\bm{T}^{-1}
$, 
%\begin{equation}
%\bm{\Sigma}_p = (\bm{T}^{-1})^\text{T}\bm{J}^{-1}_{\bm{\gamma}}\bm{T}^{-1}
%\end{equation}
where, by invoking the multivariate inverse function theorem, the inverse transformation matrix $\bm{T}^{-1}$ can be directly computed as the derivative of the location parameters with respect to the channel parameters, i.e.,
\begin{equation}
   \bm{T}^{-1} = \frac{\partial \bm{\eta}^{\text{T}}}{\partial \bm{\gamma}} =  \begin{bmatrix}
   \bar{\bm{T}}_{0,0} & \ldots & \bar{\bm{T}}_{K,0}\\
   \vdots & \ddots & \vdots \\
   \bar{\bm{T}}_{0,K} & \ldots & \bar{\bm{T}}_{K,K}
   \end{bmatrix}
\end{equation}
with each $4 \times 4$ block $\bar{\bm{T}}_{h,\ell}$, $h,\ell = 0,\ldots,K$, obtained as
\begin{equation}
\bar{\bm{T}}_{h,\ell}  \eqdef \frac{\partial \bm{\eta}_h^{\text{T}}}{\partial \bm{\gamma}_\ell} =  \begin{bmatrix}
\delta_{h \ell} & 0 & 0\\
0 & \delta_{h \ell} & 0\\
0 & 0 & \partial \bm{s}_h/\partial \tau_{\ell}\\
0 & 0 & \partial \bm{s}_h/\partial \theta_{\ell}
\end{bmatrix}.  
\end{equation}
By noting that the blocks $\bar{\bm{T}}_{h,\ell} = \bm{0}_{4 \times 4}$ for $l\geq 1$ and $h\neq\ell$, it follows that
\begin{equation}\label{eq::structureSigmap}
\bm{\Sigma}_p =  \begin{bmatrix}
\bar{\bm{T}}^{\text{T}}_{0,0} & \!\!\!\! \bm{0}_{4 \times 4} & \!\!\!\! \bm{0}_{4 \times 4} & \!\!\!\! \bm{0}_{4 \times 4} \\
\bar{\bm{T}}^{\text{T}}_{1,0} & \!\!\!\! \bar{\bm{T}}^{\text{T}}_{1,1}  & \!\!\!\! \bm{0}_{4 \times 4} & \!\!\!\! \bm{0}_{4 \times 4} \\
\bar{\bm{T}}^{\text{T}}_{2,0} & \!\!\!\! \bm{0}_{4 \times 4} & \!\!\!\! \bar{\bm{T}}^{\text{T}}_{2,2} & \!\!\!\! \bm{0}_{4 \times 4} \\
\bar{\bm{T}}^{\text{T}}_{3,0} & \!\!\!\! \bm{0}_{4 \times 4} & \!\!\!\! \bm{0}_{4 \times 4} & \!\!\!\! \bar{\bm{T}}^{\text{T}}_{3,3}
\end{bmatrix} \bm{J}^{-1}_{\bm{\gamma}} \begin{bmatrix}
\bar{\bm{T}}_{0,0} & \!\!\!\! \bar{\bm{T}}_{1,0} & \!\!\!\! \bar{\bm{T}}_{2,0} & \!\!\!\! \bar{\bm{T}}_{3,0} \\
\bm{0}_{4 \times 4} & \!\!\!\! \bar{\bm{T}}_{1,1}  & \!\!\!\! \bm{0}_{4 \times 4} & \!\!\!\! \bm{0}_{4 \times 4} \\
\bm{0}_{4 \times 4} & \!\!\!\! \bm{0}_{4 \times 4} & \!\!\!\! \bar{\bm{T}}_{2,2} & \!\!\!\! \bm{0}_{4 \times 4} \\
\bm{0}_{4 \times 4} & \!\!\!\! \bm{0}_{4 \times 4} & \!\!\!\! \bm{0}_{4 \times 4} & \!\!\!\! \bar{\bm{T}}_{3,3}
\end{bmatrix}
\end{equation}
This leads us to our first main result. 
\begin{thm}
Denoting by $\text{PEB}_{k}$ the value of the PEB when $k$ NLOS paths besides the LOS path are present, with $0 \leq k  \leq K$, then 
\begin{align}
    \text{PEB}_{K} \ge \text{PEB}_{0},
\end{align}
with equality when all paths are orthogonal.
\end{thm}
\begin{proof}
We first prove the inequality. We denote  %$\bm{\Sigma}_{\bm{\gamma}}=\bm{J}^{-1}_{\bm{\gamma}}$ and 
the $K+1$ $4\times4$ diagonal blocks of $\bm{J}^{-1}_{\bm{\gamma}}$ by $\bm{C}_0, \ldots, \bm{C}_K$. From the Schur complement, it follows that $\bm{C}_k \succeq (\bm{\Lambda}(\bm{\gamma}_k,\bm{\gamma}_k))^{-1}$.
From \eqref{eq::PEB}, it is evident that the relevant information on the estimation of $\bm{p}$ resides in the first $4 \times 4$ block of $\bm{\Sigma}_p$.  
By taking the products in \eqref{eq::structureSigmap}, it turns out that such a block is equal to $\bar{\bm{T}}^{\text T}_{0,0}\bm{C}_0\bar{\bm{T}}_{0,0}\succeq \bar{\bm{T}}^{\text T}_{0,0}(\bm{\Lambda}(\bm{\gamma}_0,\bm{\gamma}_0))^{-1}\bar{\bm{T}}_{0,0}$. Applying the PEB definition in (\ref{eq::PEB}) to both sides of this inequality, it turns out that $\text{PEB}_{K} \ge \text{PEB}_{0}$.

%Since all the entries in this block are functions of the parameters linked to the sole LOS path, we can conclude that the NLOS components have no positive impact on the PEB, i.e., they do not add useful information for the estimation of $\bm{p}$. 

We now prove the equality when paths are orthogonal. 
%To gain additional insights about the role of the NLOS components, we observe that, 
Under typical mmWave conditions, the different received paths can be resolved either in the angular or time domains, with practically negligible overlap among them. In other words, the NLOS paths can be treated as orthogonal paths carrying independent information \cite{Leitinger, AbuShaban}, leading in turn to %very small values of the (off-diagonal) cross-correlation blocks 
$\bm{\Lambda}(\bm{\gamma}_h,\bm{\gamma}_\ell) = \bm{0}_{4 \times 4}$ for $h \neq \ell$ in  \eqref{eq::FIMchannel}. 
Neglecting these terms, the approximate expression of the $\bm{\Sigma}_p$ is given by \eqref{eq::approx_CRLBp}, where now  $\bm{C}_k = (\bm{\Lambda}(\bm{\gamma}_k,\bm{\gamma}_k))^{-1}$. It then immediately follows that $\text{PEB}_{K} = \text{PEB}_{0}$. 
\end{proof}
\begin{figure*}[!t]
% ensure that we have normalsize text
\small
% Store the current equation number.
\newcounter{MYtempeqncnt}
\setcounter{MYtempeqncnt}{\value{equation}}
% Set the equation number to one less than the one
% desired for the first equation here.
% The value here will have to changed if equations
% are added or removed prior to the place these
% equations are referenced in the main text.
\setcounter{equation}{21}
\begin{equation}\label{eq::approx_CRLBp}
\bm{\Sigma}_p \approx  \begin{bmatrix}
 \bar{\bm{T}}^{\text T}_{0,0}\bm{C}_0\bar{\bm{T}}_{0,0} \!\!\!& \bar{\bm{T}}^{\text T}_{0,0}\bm{C}_0 \bar{\bm{T}}_{1,0} \!\!\!& \bar{\bm{T}}^{\text T}_{0,0}\bm{C}_0 \bar{\bm{T}}_{2,0} \!\!\!& \cdots \!\!\!& \bar{\bm{T}}^{\text T}_{0,0}\bm{C}_0 \bar{\bm{T}}_{K,0} \\ 
 \bar{\bm{T}}^{\text T}_{1,0}\bm{C}_0\bar{\bm{T}}_{0,0} \!\!\!& \bar{\bm{T}}^{\text T}_{1,0}\bm{C}_0 \bar{\bm{T}}_{1,0} + \bar{\bm{T}}^{\text T}_{1,1}\bm{C}_1 \bar{\bm{T}}_{1,1} \!\!\!& \bar{\bm{T}}^{\text T}_{1,0}\bm{C}_0\bar{\bm{T}}_{2,0} \!\!\!& \cdots \!\!\!& \bar{\bm{T}}^{\text T}_{1,0}\bm{C}_0\bar{\bm{T}}_{K,0}\\
 \bar{\bm{T}}^{\text T}_{2,0}\bm{C}_0\bar{\bm{T}}_{0,0} \!\!\!& \bar{\bm{T}}^{\text T}_{2,0}\bm{C}_0 \bar{\bm{T}}_{1,0} \!\!\!& \bar{\bm{T}}^{\text T}_{2,0}\bm{C}_0 \bar{\bm{T}}_{2,0} + \bar{\bm{T}}^{\text T}_{2,2}\bm{C}_2 \bar{\bm{T}}_{2,2} \!\!\!& \cdots \!\!\!& \bar{\bm{T}}^{\text T}_{2,0}\bm{C}_0\bar{\bm{T}}_{K,0}\\
 \vdots \!\!\!& \vdots \!\!\!& \vdots \!\!\!&  \ddots \!\!\!& \vdots\\
 \bar{\bm{T}}^{\text T}_{K,0}\bm{C}_0\bar{\bm{T}}_{0,0} \!\!\!& \bar{\bm{T}}^{\text T}_{K,0}\bm{C}_0\bar{\bm{T}}_{1,0} \!\!\!& \bar{\bm{T}}^{\text T}_{K,0}\bm{C}_0\bar{\bm{T}}_{2,0} \!\!\!&  \cdots \!\!\!& \bar{\bm{T}}^{\text T}_{K,0}\bm{C}_0 \bar{\bm{T}}_{K,0} + \bar{\bm{T}}^{\text T}_{K,K}\bm{C}_K \bar{\bm{T}}_{K,K}\\
\end{bmatrix}
\end{equation}
% Restore the current equation number.
\setcounter{equation}{\value{MYtempeqncnt}}
% IEEE uses as a separator
\hrulefill
% The spacer can be tweaked to stop underfull vboxes.
\vspace*{4pt}
\end{figure*}

This effect relates to the fact that the MS is equipped with a single-antenna receiver, hence it cannot exploit the NLOS parameters to gain additional position-related information (i.e., AOAs). This represents a major difference compared to the MIMO setup where, in general, the contribution of the NLOS components can result in a reduction of the PEB \cite{Witrisal, Arash}. 
In the MISO case, multipath propagation will degrade the MS localization only when the NLOS paths and the LOS significantly overlap, but will never improve the PEB compared to the LOS-only case. From such considerations, it also follows that in the MISO setup (i) localization without LOS is not possible; and (ii) the NLOS paths cannot be used to synchronize the MS to the BS. Both aspects are in contrast to the MIMO case, as shown in \cite{Henk5G}. 

Remarkably, we observe from \eqref{eq::approx_CRLBp} that mapping of the scatterers positions is still possible in spite of the fact that the receiver has only a single antenna, that is, it cannot perform any spatial processing. More specifically, the terms in the main diagonal of \eqref{eq::approx_CRLBp} reveal that the accuracy in the estimation of each scatterer's position $\bm{s}_k$ is linked to the parameters of the associated $k$-th NLOS path $\bar{\bm{T}}^{\text T}_{k,k}\bm{C}_k\bar{\bm{T}}_{k,k}$, as well as to the parameters related to the LOS link $\bar{\bm{T}}^{\text T}_{k,0}\bm{C}_0\bar{\bm{T}}_{k,0}$. Given the additive nature of such terms, the lower the uncertainty in the LOS parameters, the higher the accuracy in mapping the multipath environment. 

% \begin{equation}
% \!\!\!\mathrm{CRLB}_p \approx \begin{bmatrix}
% \sum\limits_{k=0}^{K} \bar{\bm{T}}^T_{0,k}\bm{C}_k\bar{\bm{T}}_{0,k} \!\!\!& \bar{\bm{T}}^T_{0,1}\bm{C}_1 \bar{\bm{T}}_{1,1} \!\!\!& \cdots \!\!\!& \bar{\bm{T}}^T_{0,K}\bm{C}_K \bar{\bm{T}}_{K,K} \\ 
%  \bar{\bm{T}}^T_{1,1}\bm{C}_1\bar{\bm{T}}_{0,1} \!\!\!& \bar{\bm{T}}^T_{1,1}\bm{C}_1 \bar{\bm{T}}_{1,1} \!\!\!&  \cdots \!\!\!& \bm{0}_{4 \times 4}\\
%  \vdots \!\!\!& \vdots \!\!\!&  \ddots \!\!\!& \vdots\\
%  \bar{\bm{T}}^T_{K,K}\bm{C}_K\bar{\bm{T}}_{0,K} \!\!\!& \bm{0}_{4 \times 4} \!\!\!&  \cdots \!\!\!& \bar{\bm{T}}^T_{K,K}\bm{C}_K\bar{\bm{T}}_{K,K}\\
% \end{bmatrix}
% \end{equation}

%The additive nature of such terms implies that the presence of NLOS paths is detrimental to the estimation of the MS position $\bm{p}$, as they increase the CRLB compared to the LOS-only case.

%implies that the presence of NLOS paths is detrimental to the estimation of the MS position $\bm{p}$, as they increase the CRLB compared to the LOS-only case.

%However, in the simulation analysis conducted in Sec. IV, we will show that accurate localization performance can be obtained also in the considered mmWave MISO setup, despite the increased number of nuisance parameters due to multipath propagation. 

\stepcounter{equation}
\section{Joint Maximum Likelihood Localization and Mapping}

In this section, we present the joint maximum likelihood (ML) estimator, a low-complexity channel estimator working in two dimensions, and the localization and mapping algorithm. We also show that the ML estimator can be performed equivalently in the position domain, and provide insights into the obtained solutions.
\subsection{Maximum Likelihood Estimation of Channel Parameters}\label{sec:jointML_channel}
%In this section, we derive the joint maximum likelihood (ML) estimator of the unknown position-related channel parameters, which will be used to retrieve the MS position as well as to map the multipath environment.
We start the derivation by noting that each received signal $y^{g}[n]$, $1 \leq g \leq G$, $0 \leq n \leq N-1$, can be statistically characterized as 
\begin{equation}
y^{g}[n] \sim \mathcal{CN}(\sqrt{N_{\text{\tiny BS}}}  \bar{\bm{h}}^{\text{T}}[n]\bm{z}^{g}[n], \sigma^2)
\end{equation}
where $\bar{\bm{h}}^{\text{T}}[n] = \sum_{k=0}^K \alpha_k\e^{\frac{-j2\pi n \tau_k}{NT_S}}\bm{a}^{\text H}_{\text{\tiny BS}}(\theta_k)$ and all the parameters are treated as deterministic unknowns, except the transmitted symbols $\bm{z}^{g}[n]$, which are assumed known to the receiver, and the number of paths $K$, which can be determined during the initial access phase \cite{Barati, Giordani}. To formulate the estimation problem, we re-order  the unknown parameters as $\bm{\varphi} =[\bm{\Theta}^{\text{T}} \  \bm{\psi}^{\text{T}}]^{\text{T}}$, 
%\begin{equation}
%\bm{\varphi} =[\bm{\Theta}^{\text{T}} \  \bm{\psi}^{\text{T}}]^{\text{T}}
%\end{equation} 
where $\bm{\Theta} = [\theta_0 \ \tau_0 \ \cdots \ \theta_K \ \tau_K]^{\text{T}}$ represents the parameters of interest linked to the desired MS and scatterers positions, while $\bm{\psi} = [\sigma^2\ \bm{\alpha}^{\text{T}}]^{\text{T}}$ with $\bm{\alpha} = [\alpha_0 \ \cdots \ \alpha_K]^{\text T}$ denotes the vector of nuisance parameters. Notice that, differently from the LOS-only scenario, including the NLOS links in the localization process introduces additional unknown parameters that make the resulting estimation problem much more challenging. Following the ML criterion, the estimation problem can be thus formulated as
\begin{align}\label{eq::MLIPE1}
\hat{\bm{\Theta}} = \mathrm{arg} \max_{\bm{\Theta}} [\max_{\bm{\psi}} L(\bm{\Theta}, \bm{\psi})]
\end{align}
where $L(\bm{\Theta}, \bm{\psi}) \eqdef \log f(\bm{Y}|\bm{\Theta}, \bm{\psi})$ and $f(\cdot)$ denotes the probability density function of the observations $\bm{Y}$ given $\bm{\psi}$ and $\bm{\Theta}$.
A more convenient rewriting of the channel model in \eqref{eq::fullchannelmodel} allows us to express the likelihood in \eqref{eq::MLIPE1} 
as
\begin{align}\label{eq::IPEloglike_1}
L(\bm{\Theta},\bm{\psi}) &= - NG\log(\pi \sigma^2) \nonumber\\
& - \frac{1}{\sigma^2}\sum_{g=1}^G \|\bm{y}^{g} - \sqrt{N_{\text{\tiny BS}}}  \bm{Q}^{g}\bm{\alpha}\|^2
\end{align}
where
\begin{equation}
\bm{Q}^{g} = \begin{bmatrix} (\bm{z}^g[0])^{\text T}\bm{D}[0] \\ \vdots \\ (\bm{z}^g[N-1])^{\text T}\bm{D}[N-1]
\end{bmatrix}  \in \mathbb{C}^{N \times (K+1)}  
\end{equation}
with $\bm{y}^{g} = [y^g[0] \, y^g[1] \, \cdots \, y^g[N-1]]^\text{T}$ the $g$-th column of the observation matrix $\bm{Y}$, and
\begin{equation}
\bm{D}[n] =
\begin{bmatrix}
\e^{\frac{-j2\pi n \tau_0}{NT_S}}\bm{a}^{*}_{\text{\tiny BS}}(\theta_0) & \cdots & \e^{\frac{-j2\pi n \tau_K}{NT_S}}\bm{a}^{*}_{\text{\tiny BS}}(\theta_K)
\end{bmatrix}.
\end{equation}
It is easy to observe that the noise variance can be estimated as $\hat{\sigma}^2 = \sum_{g=1}^G \|\bm{y}^{g} - \sqrt{ N_{\text{\tiny BS}}} \bm{Q}^{g}\bm{\alpha}\|^2 / (NG)$, leading to the compressed likelihood
\begin{align}
\label{eq::IPEloglike_2}
L_K(\bm{\Theta},\bm{\alpha}) = \sum_{g=1}^G \|\bm{y}^{g} - \sqrt{N_{\text{BS}}} \bm{Q}^{g}\bm{\alpha}\|^2
\end{align}
where $L_K(\bm{\Theta},\bm{\alpha})$ is the compressed negative log-likelihood function in presence of $K$ NLOS paths. %The joint ML estimation problem thus reduces to
%\begin{equation}\label{eq::MLIPE2}
%\hat{\bm{\Theta}} = \mathrm{arg} \min_{\bm{\Theta}} [\min_{\bm{\alpha}} L_K(\bm{\Theta}, \bm{\alpha})].
%\end{equation}
Eq. \eqref{eq::IPEloglike_2} can be optimized with respect to the entire vector $\bm{\alpha} \in \mathbb{C}^{(K+1) \times 1}$, yielding 
%\begin{align}
$\hat{\bm{\alpha}} = {1}/{\sqrt{N_{\text{BS}}}} \bm{Q}^{-1}\sum_{g=1}^G(\bm{Q}^{g})^{\text{H}}\bm{y}^{g}$
%\end{align}
%\begin{equation}
%\hat{\bm{\alpha}} = \frac{1}{\sqrt{N_{\text{BS}}}} \bm{Q}^{-1}\sum_{g=1}^G(\bm{Q}^{g})^{\text{H}}\bm{y}^{g}
%\end{equation}
where $\bm{Q} = \sum_{g=1}^G (\bm{Q}^g)^{\text{H}}\bm{Q}^g 
$. 
%\begin{equation}
%\bm{Q} = \sum_{g=1}^G (\bm{Q}^g)^{\text{H}}\bm{Q}^g.   
%\end{equation}
Substituting these minimizing values back in \eqref{eq::IPEloglike_2} leads to 
\begin{equation}\label{eq::IPEloglike_final}
L_K(\bm{\Theta}) = \sum_{g=1}^G \|\bm{y}^{g} -  \bm{Q}^{g}(\bm{\Theta})\hat{\bm{\alpha}}(\bm{\Theta})\|^2
\end{equation}
and, accordingly, the final joint ML estimator is given by
\begin{equation}\label{eq::finalJML}
\hat{\bm{\Theta}} =  \mathrm{arg} \min_{\bm{\Theta}}  L_K(\bm{\Theta}).
\end{equation}
The cost function  \eqref{eq::IPEloglike_final} is highly non-linear in the $2(K+1)$ unknown parameters and does not admit a closed-form solution or a multidimensional exhaustive search.  
%Given the high dimension of the optimization space, approaching the resolution of \eqref{eq::finalJML} with a multidimensional exhaustive search appears an unfeasible task, especially for the limited processing capabilities available at the MS, even though $K$ is typically very small. 
Therefore, the joint ML need to be solved by resorting to iterative numerical optimization routines such as, for instance, those based on derivatives of the cost function (e.g., gradient descent or its variants \cite{NumReci}), or by employing more direct approaches such as the Nelder-Mead method \cite{Nelder}, starting from a good initial estimate. 
%It is worth noting that, for all the above methods, a good initialization of the  parameters vector $\bm{\Theta}$ is crucial  to reach a solution close to the optimal. Indeed, since the cost function in \eqref{eq::IPEloglike_final} is non-convex, a random initialization of $\bm{\Theta}$ will likely bring  iterative numerical routines to get trapped into local minima, thus yielding large estimation errors. 

\subsection{Low-complexity Channel Parameter Estimation}
To solve the channel parameter estimation problem in practice, we take advantage of the sparse nature of the mmWave channel and propose a reduced-complexity suboptimal approach to obtain a good initial estimate of $\bm{\Theta}$. The main idea consists in exploiting the fact that  the received paths are almost orthogonal between each other, as discussed in Sec. \ref{sec::roleNLOS}. Under this assumption, the joint ML estimation problem can be approximated to a problem of multiple single-path estimation, where  each path can be described by the following simplified channel model
\begin{equation}\label{eq::spModel}
\tilde{\bm{h}}^\text{T}[n] = \alpha 
e^{\frac{-j2\pi n \tau}{NT_S}}\bm{a}_{\text{\tiny BS}}^\text{H}(\theta), \quad n=0,\ldots,N-1
\end{equation}
with $\alpha$, $\theta$ and $\tau$ complex amplitude, AOD and TOF of a single path, respectively. Replacing \eqref{eq::fullchannelmodel} with \eqref{eq::spModel} in the derivation of the joint ML immediately leads to the cost function of the single-path ML estimator, denoted as
%\begin{equation}\label{eq::SP_channelesti}
%[\hat{\theta}\ \hat{\tau}]  = \argmin_{  \theta, \tau} \tilde{L}(\theta,\tau)
%\end{equation}
 $L_0(\theta,\tau)$, whose expression can be straightforwardly obtained as a special case of \eqref{eq::IPEloglike_final} for $\bm{\Theta} = [\theta \ \tau]^\text{T}$, which is tantamount to considering $K = 0$ in the original joint ML estimation problem. 
 
 In analogy to traditional subspace-based AOA estimation, we leverage  orthogonality among the paths and interpret $L_0(\theta,\tau)$ as a kind of ``pseudospectrum'', whose minima occur in correspondence of pairs ($\theta$, $\tau$)  close to the actual channel parameters $\theta_k$ and $\tau_k$ of each $k$-th path. 
 As it will be shown in Sec. \ref{sec::simanalysis}, searching for the $K+1$ dominant minima in $L_0(\theta,\tau)$ and using these as initial estimates in 
the iterative minimization of $L_K(\bm{\Theta})$  
%procedure initialized with an estimate $\hat{\bm{\Theta}}$ obtained by searching for the $K+1$ dominant minima in $L_0(\theta,\tau)$ 
can efficiently solve the joint ML estimation problem and attain the theoretical performance bounds, but at the significantly reduced cost of a coarse two-dimensional search over the space $(\theta,\tau)$ instead of a prohibitive $2(K+1)$-dimensional search.

%Such coarse knowledge can then be used in a second step to properly initialize a $2(K+1)$-dimensional gradient descent optimization, which as it will be shown can efficiently solve the joint ML problem and achieve the corresponding theoretical bounds.

%by replacing the unknown scatterers' positions $\bm{s}_k$ with their estimates $\hat{\bm{s}}_k$ obtained in the previous localization step. In doing so, we can solve \eqref{eq::DPEloglike_final} only for  the unknown $\bm{p}$, thus reducing the computational burden to a two-dimensional  search.

\subsection{Localization and Mapping}
From the theoretical analysis conducted in Sec. \ref{sec::roleNLOS}, it emerged that in a MISO setup, NLOS components cannot be harnessed to determine the unknown MS position. In this respect, the natural way to obtain an estimate of $\bm{p}$ is to exploit the sole LOS position-related parameters, which can be identified among the $K+1$ estimated pairs $(\hat{\theta}_k,\hat{\tau}_k)$ as the pair with the minimum value for $\hat{\tau}_k$, while the remaining pairs are used for determining the map. 

\subsubsection{Localization}
In the following, we will refer to such estimates as $(\hat{\theta}_{\text{\tiny LOS}},\hat{\tau}_{\text{\tiny LOS}})$. The unknown MS position can be then determined by solving \eqref{eq:tau0}--\eqref{eq:theta0} for $\bm{p}$: 
\begin{equation}\label{eq::estMSpos}
\hat{\bm{p}} = c\hat{\tau}_{\text{\tiny LOS}} [\cos \hat{\theta}_{{\text{\tiny LOS}}} \ \sin \hat{\theta}_{{\text{\tiny LOS}}}]^\text{T}.   
\end{equation}

\subsubsection{Mapping}
Once the estimate $\hat{\bm{p}}$ is obtained, it can be used in conjunction with each pair $(\hat{\theta}_k,\hat{\tau}_k)$, $k \geq 1$, to retrieve the related scatterer's position $\bm{s}_k$. More precisely, the direction $\hat{\theta}_k$ constrains the sought $\hat{\bm{s}}_k$ to lie on the straight line passing by the BS position and having angular coefficient $\tan(\hat{\theta}_k)$. %$m(\hat{\theta}_k) = \tan(\hat{\theta}_k)$. 
Among all the possible candidate positions on that line, we select as $\hat{\bm{s}}_k$ the one satisfying the distance constraint $\hat{d}_k = c\hat{\tau}_k$ with $\hat{\bm{p}}$.
%In formulas, we have that
%\begin{align}
 %&\sqrt{\hat{s}^2_{k,x}(1+m(\hat{\theta}_k)^2)} \nonumber \\  &+ \sqrt{\hat{s}^2_{k,x}(1+m(\hat{\theta}_k)^2 - 2(\hat{p}_x + m(\hat{\theta}_k)\hat{p}_y)\hat{s}_{k,x} + \hat{p}^2_x + \hat{p}^2_y} = \hat{d}_k
%\end{align}
After straightforward steps, we find
%with $\hat{s}_{k,y} = m(\hat{\theta}_k)\hat{s}_{k,x}$.
%It is a simple matter to verify that the arguments of both square roots are always positive. 
(assuming for simplicity that $\hat{s}_{k,x} > 0 \;\forall k=1,\ldots,K$) that the position of the $k$-th scatterer can be estimated in closed-form as
\begin{equation}
  \begin{cases}
\hat{s}_{k,x} =\frac{1}{2} \frac{(c\hat{\tau}_k)^2 - \hat{p}^2_x - \hat{p}^2_y}{\sqrt{1 + \mathrm{tan}^2(\hat{\theta}_k)} \, c\hat{\tau}_k - \hat{p}_x - \mathrm{tan}(\hat{\theta}_k)\hat{p}_y}\\
\hat{s}_{k,y} = \mathrm{tan}(\hat{\theta}_k)\hat{s}_{k,x}
\end{cases}\label{eq::est_scatt_pos}.
\end{equation}
We observe that, in line with the theoretical findings in Sec. \ref{sec::roleNLOS}, the accuracy on the estimation of $\bm{s}_k$ depends on the quality of the NLOS parameters estimates $\hat{\theta}_k$ and $\hat{\tau}_k$, as well as on the goodness of $\hat{\bm{p}}$, which is estimated based on the LOS-only parameters $\hat{\tau}_{\text{\tiny LOS}}$ and $\hat{\theta}_{\text{\tiny LOS}}$. Given the nonlinear nature of the geometric estimator \eqref{eq::est_scatt_pos}, it is not trivial to figure out how the involved parameters  impact the  estimation of $\hat{\bm{s}}_k$.

\subsection{Equivalence of Maximum Likelihood Estimation in Channel and Position Domains}\label{sec:jointML_posdomain}

%In the preceding subsections, we have derived the joint ML estimator of the  channel parameters and showed how the obtained estimates $\hat{\bm{\Theta}}$  can be used to determine the MS position as well as to map the surrounding  environment. In this respect, it is important to remark that there is no loss of optimality in this procedure, being the closed-form estimators \eqref{eq::estMSpos} and \eqref{eq::est_scatt_pos} derived from mere geometric relationships between the channel and location parameters. 
In this section, we briefly discuss an alternative formulation of the  ML estimation problem in the position domain, showing that the resulting estimator is equivalent to the one in the channel domain. 
Without loss of generality, we focus on the single-path cost function $L_0(\theta,\tau)$, but the same reasoning can be easily applied also to  $L_K(\bm{\Theta})$. More precisely, by expressing the channel parameters $\theta$ and $\tau$ as a function of their corresponding location parameters according to \eqref{eq:tau0}--\eqref{eq:theta0},  \eqref{eq::spModel} can be equivalently rewritten as
\begin{equation}\label{eq::spModel_position}
\tilde{\bm{h}}^\text{T}[n] = \alpha 
e^{\frac{-j2\pi n \|\bm{s}\|}{cNT_S}}\bm{a}_{\text{\tiny BS}}^\text{H}(\mathrm{atan2}(s_y,s_x)), \quad n=0,\ldots,N-1
\end{equation}
where $\bm{s} = \bm{g}(\theta,\tau) = c\tau [\cos \theta \ \sin \theta]^\text{T}$ and $\bm{g}(\cdot)$ is the bijective mapping (transformation from polar to Cartesian coordinates) between the channel and position parameters.
Accordingly, we denote by $L_0(\bm{s})$ the position-domain counterpart of the cost function $L_0(\theta,\tau)$. 
The equivalence between the ML formulations in both channel and position domains easily follows by observing that, given the bijective mapping $\bm{s} = \bm{g}(\theta,\tau)$, the likelihood $L_0(\theta,\tau)$ can be rewritten as a function of $\bm{s}$, i.e.,
\begin{equation}
L_0(\theta,\tau) = L_0(\bm{g}^{-1}(\bm{s})). 
\end{equation}
% Let us assume that the latter is minimized over a given two-dimensional finite set $\mathcal{X}$ of  $(\theta,\tau)$ pairs. Then, it immediately follows that
% \begin{equation}
% \bm{g}\left(\argmin_{\mathcal{X}}L_0(\theta,\tau)\right) =   \argmin_{\mathcal{Y}} L_0(\bm{g}^{-1}(\bm{s})) 
% \end{equation}
% with the finite set $\mathcal{Y} = \bm{g}(\mathcal{X})$, hence the ML estimators in the channel and position domains are perfectly equivalent. 

Analogously to $L_0(\theta,\tau)$, searching for the $K+1$ dominant minima of $L_0(\bm{s})$  provides an initial estimate of the MS and scatterers positions, which can be subsequently used to solve the joint ML problem in the position domain. In this respect, it is worth noting that since the single-path model \eqref{eq::spModel_position} is unable to capture the geometric reflections of the propagating rays, the bijective transformation $\bm{g}(\cdot)$ will map the TOF $\tau_k$ of each NLOS path into a position that falls within a distance $d_k = c\tau_k$ from the BS, along a direction identified by the AOD $\theta_k$, thus leading to a final position that does not coincide with the actual position of the $k$-th scatterer. Let us denote by $\bm{s}^e_k = [s^e_{k,x} \ s^e_{k,y}]^\text{T}$, $k \geq 1$, such 
``equivalent" positions. Then, each $\bm{s}^e_k$
%a first initial  estimate of the equivalent vector $\bm{\Theta}^e_p = [\bm{p}^T \ (\bm{s}^e_1)^T \ \cdots \ (\bm{s}^e_K)^T]^T$ can be obtained by searching for the $K+1$ dominant minima in $\tilde{L}(\bm{s})$. In a second step, 
can be mapped back to its corresponding position $\bm{s}_k$ by applying some  geometric considerations: first, we write the parametric expression for  the segment passing by the BS and the equivalent position $\bm{s}^e_k$, that is, $\bm{s}_k(\lambda) = \lambda \bm{p}_{\text{\tiny BS}} + (1-\lambda)\bm{s}^e_k$,  $\lambda \in [0,1]$. Then, we retain as position $\bm{s}_k$ the point on the line corresponding to the value $\lambda^*$ satisfying $\|\bm{s}^e_k - \bm{s}_k(\lambda^*)\| = \|\bm{s}_k(\lambda^*) -\bm{p}\|$, that is
\begin{equation}
\lambda^* = \frac{1}{2}\frac{\|\bm{s}^e_k - \bm{p}\|^2}{\|\bm{s}^e_k \|^2 - (p_x s^e_{k,x} + p_y s^e_{k,y})}
\end{equation}
(where without loss of generality we kept assuming that the BS is placed at the origin of the reference system). Given the  equivalence between the two estimators, in the following we present the results only for one of them; in particular, we opt for the joint ML in the channel domain, being closer to the physics of the channel hence more easily interpretable in terms of paths (i.e., angles and delays).

\section{Simulation Analysis and Results}\label{sec::simanalysis} 
In this section, we present a simulation analysis aimed at evaluating the performance of the proposed joint ML estimator, considering different values of the relevant parameters, also in comparison with the theoretical bounds derived in Sec.~\ref{sec:bounds}. To evaluate the performance, we consider the Root Mean Squared Error (RMSE) estimated on the basis of 1000 Monte Carlo independent trials.

\subsection{Simulation Setup}
The analyzed scenario consists of a single BS equipped with $N_{\text{\tiny BS}} = 20$ antennas, placed at a known position $\bm{p}_{\text{\tiny BS}}= [3 \ 0]^\text{T}$ m, while the MS is located at $\bm{p} = [10 \ 4]^\text{T}$ m. 
The localization process is carried out by exploiting  only a single broadcast signal in DL ($G = 1$) with bandwidth $B = 40$ MHz over a central frequency $f_c = 60$ GHz, using $N = 20$ different subcarriers. The simulations are carried out without assuming any a priori knowledge of the MS and scatterers positions; accordingly, we set the beamforming matrix $\bm{F}^g[n]$ to have $M = N_{\text{\tiny BS}}/2$ uniformly-spaced beams that cover the whole considered area, and keep it constant over each transmission $g$ and subcarrier $n$. If some knowledge about the environment is available, it can be incorporated in the beamforming matrix as suggested in \cite{ICASSP2020}.

We compute the channel path loss $\rho_k$ for each $k$-th path according to the geometry statistics in \cite{Li1,Li2}. For $k=0$, we retrieve the path loss experienced by the LOS path as $1/\rho_0 = \xi^2(d_0) \left({\lambda_c}/{(4 \pi d_0)} \right)^2 $, 
%\begin{equation}
%1/\rho_0 = \xi^2(d_0) \left(\frac{\lambda_c}{4 \pi d_0} \right)^2
%\end{equation}
with $\xi^2(d_0)$ modeling the effects of the atmospheric attenuation at $d_0$, while the remaining factor is the well-known free space loss at a distance $d_0$. According to the experimental campaigns discussed in \cite{rappaport_5G}, we set $\xi^2(d_0)$ to $16$ dB/Km. The  channel gain is computed as $h_0 = a_0 \e^{j \varphi_0}$ with $a_0 = \sqrt{P_t}$ the amplitude related to the transmitted power $P_t$ and $\varphi_0$ the corresponding phase.

Consistently with the typical mmWave channel characteristics, we assume that each NLOS path is generated from a single dominant reflector \cite{Li2}; therefore, the path loss $\rho_k$ ($k \geq 1$) along the $k$-th NLOS link is evaluated according to $1/\rho_k = \omega\Omega(d_k) \left({\lambda_c}/{(4 \pi d_k)} \right)^2$, 
%\begin{equation}
%1/\rho_k = \omega\Omega(d_k) \left(\frac{\lambda_c}{4 \pi d_k} \right)^2
%\end{equation}
with $d_k$ total length of the path and $\Omega(d_k) = (\gamma_r d_k)^2\e^{-\gamma_r d_k}$ the Poisson distribution with density $\gamma_r$ modeling the geometry of the environment. Following the specifications in \cite{Li2}, we set $\gamma_r = 1/7$; as to $\omega$, it models the first-order reflection effects and depends on the specific characteristics of the propagation environment. To assess the algorithm performance under different operating conditions, we set $\omega$ in order to obtain a varying range of power for each NLOS path, the latter expressed in terms of LOS-to-multipath ratio (LMR) $\text{LMR}_k = {P_{\text{\tiny LOS}}}/{ P^k_{\text{\tiny NLOS}}} = {\rho_k}/{\rho_0}$.   
%\begin{equation}
%\text{LMR}_k = \frac{P_{\text{\tiny LOS}}}{ P^k_{\text{\tiny NLOS}}} = \frac{\rho_k}{\rho_0}.
%\end{equation}
This indicator  reflects the theoretical insights provided by the CRLB analysis in Sec.~\ref{sec:bounds}: indeed, from the diagonal elements of \eqref{eq::approx_CRLBp}, we observed that the lower bounds on the estimation of each scatterer position $\bm{s}_k$ depend only on the parameters of the associated $k$-th NLOS path and on the parameters of the LOS path, hence the ratio between their powers represents a meaningful parameter to discriminate between favorable (i.e., stronger LOS) and unfavorable (i.e., weaker LOS) conditions.
Accordingly, the total LMR is given by
%\begin{equation}
$\text{LMR} = {P_{\text{\tiny LOS}}}/{\sum_{k=1}^K P^k_{\text{\tiny NLOS}}} = {1/\rho_0}/{\sum_{k=1}^K 1/\rho_k}$. 
%\end{equation}
The transmit power $P_t$ adopted by the BS is varied (from about $0.1$ mW up to about $10$ mW) in order to obtain different ranges of SNR, defined as $\text{SNR} \eqdef 10\log_{10}\left({P_t}/{(\rho_0 N_0B)}\right)$, 
%\begin{equation}
%\text{SNR} \eqdef 10\log_{10}\left(\frac{P_t}{\rho_0 N_0B}\right)
%\end{equation}
where %$\log_{10}(\point)$ denotes the base-10 logarithm and 
$N_0$ is the noise power spectral density. %is the receiver noise figure, i.e., $N_0B = k_BT_0B$, $k_B$ being  the Boltzmann constant and $T_0$ the standard thermal noise temperature.

\begin{figure}
\centering
 \includegraphics[width=0.55\textwidth]{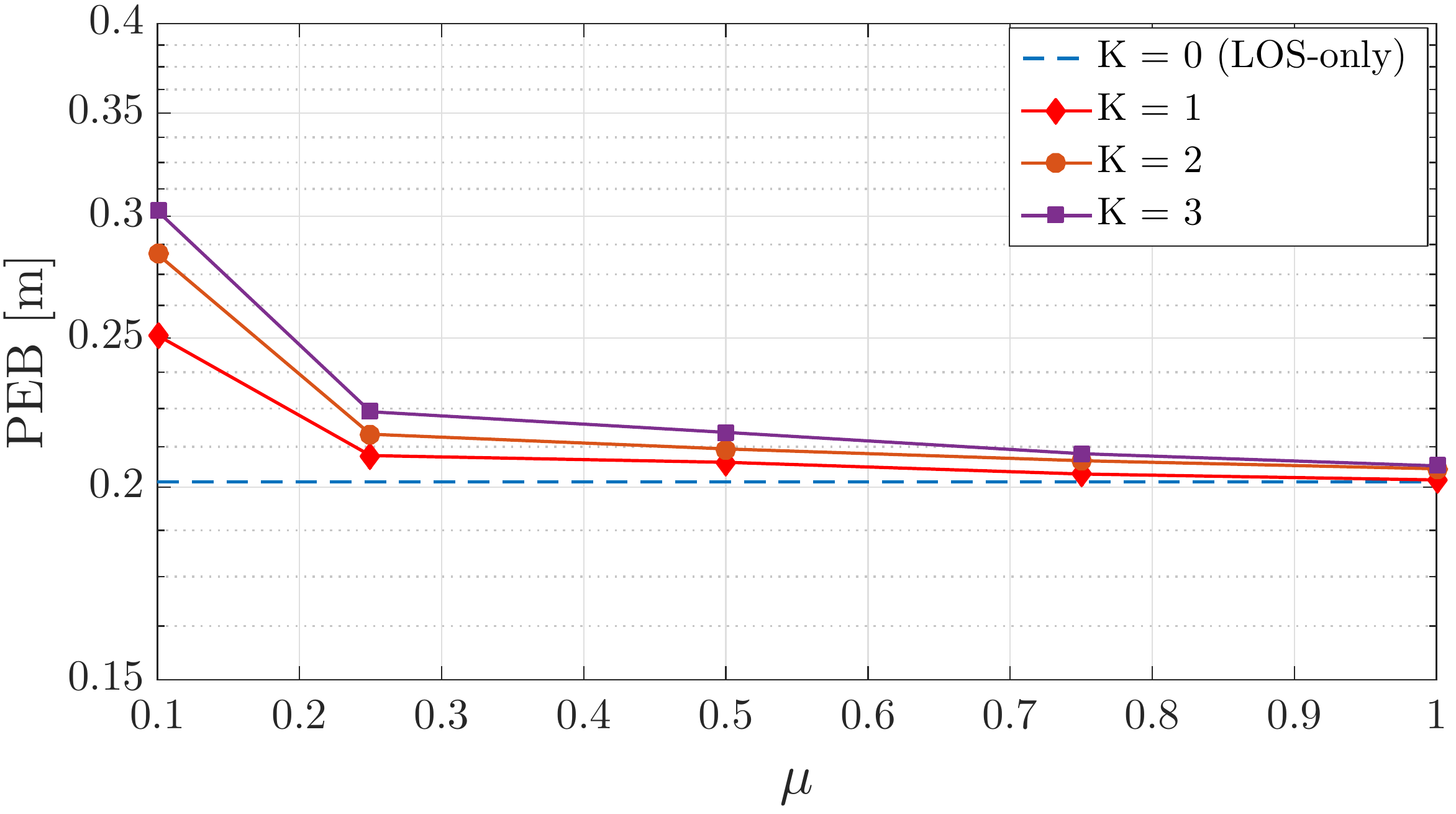}
 	\caption{PEB on the estimation of $\bm{p}$ as a function of the degree of separation among LOS and NLOS paths $\mu$, for a varying number of NLOS paths.}
\label{fig:PEB}
 \end{figure}

\subsection{Results and Discussion}
\subsubsection{Analysis of the theoretical bounds}we start the analysis by providing a numerical interpretation of the relevant CRLBs derived using the FIM analysis in Sec. \ref{sec:bounds}. In Fig.~\ref{fig:PEB}, we investigate the achievable theoretical accuracy  in the estimation of the desired MS position $\bm{p}$ in presence of a number of NLOS paths $K$ varying from a minimum of zero (i.e., LOS-only scenario) up to a maximum of three\footnote{Recent measurement campaigns conducted at mmWave frequencies have shown that, due to severe path-loss and frequent blockages, a realistic channel typically consists of a very small number of dominating paths \cite{Rappaport}.}. To reproduce different geometric configurations of the environment, we fix three reference directions from the MS to the scatterers to $-20^{\circ}$, $50^{{\circ}}$ and $70^{{\circ}}$, respectively, and vary each position $\bm{s}_k$ $(k \geq 1$) along its corresponding direction in order to obtain a distance between the MS and the $k$-th scatterer equal to $d_{k,2} = \ell_{k}\mu$, with $\mu \in (0,1]$ a scaling parameter introduced to increase or decrease the degree of separation (in terms of both TOF and AOD) among the LOS path and the NLOS paths. The three reference distances are set to $\ell_{1} = 20$ m, $\ell_{2} = 28$ m, and $\ell_{3} = 36$ m, so that the resulting scenario is compatible with the expected coverage in mmWave 5G systems \cite{Koivisto}. 
In agreement with the theoretical findings in Sec. \ref{sec::roleNLOS}, we observe that the PEB does not experience significant changes as $K$ increases, confirming that the estimation of $\bm{p}$ is not harmed by the presence of the additional NLOS paths at the receive side. This behavior also confirms the orthogonality among the different paths, since the NLOS paths start to have a noticeable effect only for very small values of $\mu$.
%in either less separated ($\mu < 0.5)$ or more separated ($\mu \geq 0.5$) scenarios. 
Furthermore, the very slight differences among the PEB curves reveal that the residual reciprocal interference among the paths is mainly linked to the overall multipath power and it is otherwise independent from the effective number of NLOS paths $K$. Therefore, to ease the presentation and without loss of generality, in the following we stick to the case of a single NLOS path (i.e., $K = 1$) and evaluate the proposed algorithm performance for different values of the multipath power. Finally, notice that the values assumed by the PEB demonstrate that cm-level localization accuracy can be  achieved in the considered mmWave MISO setup, in spite of the fact that the receiver can only exploit a single antenna to cope with multipath propagation.

 \begin{figure}
     \centering
      \includegraphics[width=0.55\textwidth]{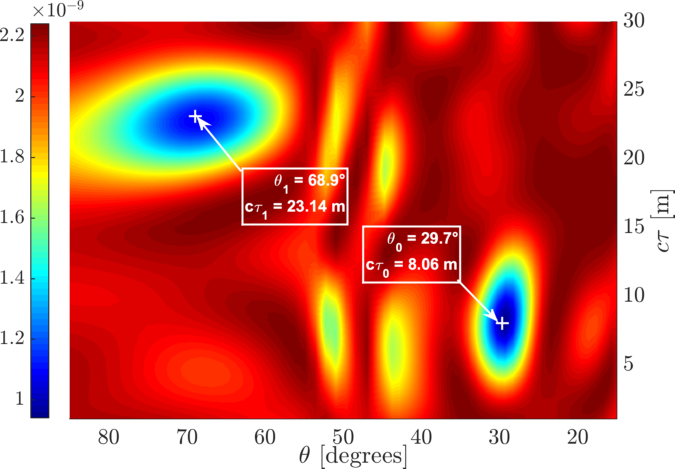}
 	\caption{Possible evaluation of the cost function $L_0(\theta,\tau)$ for a scenario with $K=1$ NLOS path.} 
\label{fig:costSPML}
 \end{figure}

\subsubsection{Comparison between channel domain and position domain estimation}the approach proposed in Sec.~\ref{sec:jointML_channel} originates from the idea that a first initial estimate of the unknown vector $\bm{\Theta}$ can be obtained by searching for the $K+1$ dominant minima in the single-path cost function $L_0(\theta,\tau)$. To validate such an intuition, in Fig.~\ref{fig:costSPML} we report a graphical example of a possible evaluation of $L_0(\theta,\tau)$ over a discrete two-dimensional $64 \times 64$ grid\footnote{Notice that we are considering a fine grid (with resolution 0.2 m in range and 0.5 degrees in angle) only for the sake of visualization; for practical implementation, this is not necessary since, as shown later, a very coarse grid (with resolution 2 m in range and 4 degrees in angle) is sufficient to achieve best performance while keeping low the complexity.} of ($\theta,\tau)$ pairs. The simulation is conducted assuming a single scatterer placed at $\bm{s}_1 = [8 \ 13]^\text{T}$ m, which generates a reflected NLOS path having a power 3 dB less than the LOS, at $\text{SNR}=5$ dB.
As expected, the cost function is highly non-linear and exhibits several local minima. However, the two dominant minima are in the neighborhood of the actual $(\theta_k,\tau_k)$ pairs, $k=0,1$ (indicated by crosses), meaning that $L_0(\theta,\tau)$ is able to capture the angular and time ``signature" of each individual path, although with an accuracy that worsens for lower SNR and less separable paths. This is remarkable since  $L_0(\theta,\tau)$ is a suboptimal function that ignores the presence of more than one path in the received signal $\bm{Y}$.

For the sake of comparison, in Fig.~\ref{fig:costSPML_position} we report the evaluation of the position-domain cost function $L_0(\bm{s})$ (see Sec.~\ref{sec:jointML_posdomain}) over a discrete grid of  $(s_x,s_y)$ pairs obtained from the previous $64 \times 64$ grid (in the channel domain) by applying a Cartesian transformation to each  $(\theta,\tau)$ pair  (i.e.,  $[s_x \ s_y]^\text{T} = c\tau [\cos \ \theta \sin \theta]^\text{T}$ pair), while the remaining parameters are set as in Fig.~\ref{fig:costSPML}.
\begin{figure}
     \centering
      \includegraphics[width=0.55\textwidth]{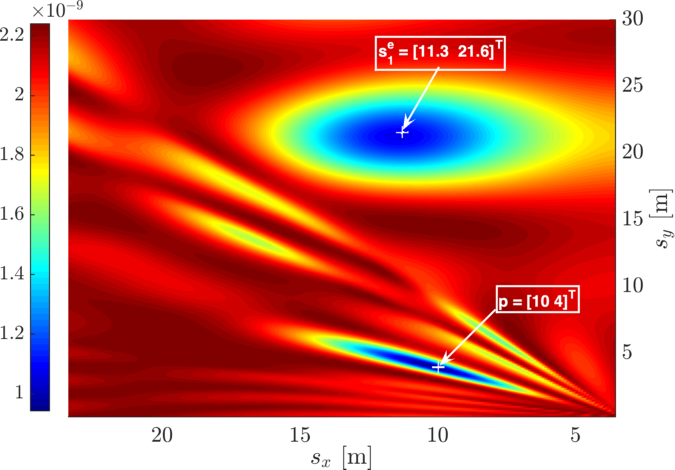}
 	\caption{Possible evaluation of the cost function $L_0(\bm{s})$ for a scenario with $K=1$ NLOS path.} 
\label{fig:costSPML_position}
 \end{figure}
We observe that, also in the location domain, there are two dominant minima that clearly emerge in evaluating the cost function $L_0(\bm{s})$. In line with the discussion in Sec.~\ref{sec:jointML_posdomain}, the bottom-most minimum occurs in the neighborhood of the actual MS position $\bm{p}$, while the other one (linked to the NLOS parameters) is located in the vicinity of the equivalent scatterer position $\bm{s}^e_1 = [11.3 \ 21.6]^\text{T}$ m.\footnote{It is worth noting that a naive search of the first $K+1$ minima in any single-path cost function would likely produce erroneous estimates of the sought channel or position parameters, respectively: in fact, since each dominant minimum is quite spread (blue areas in both figures), the search would likely lead to incorrectly selecting multiple local minima belonging to the neighborhood of the same dominant minimum. To overcome such a drawback, one can resort to the well-known space-alternating generalized expectation-maximization (SAGE) method, which sequentially estimate each  $(\theta_k,\tau_k)$ pair and compensate its contribution before searching for the next dominant minimum (i.e., the next  $(\theta_k,\tau_k)$ pair) in the cost function. This approach, theoretically introduced in \cite{Sage}, has been extensively applied for parameter extraction from extensive channel measurement data \cite{SageApp1,SageApp2}.}

\begin{figure}
    \centering
    \begin{subfigure}[t]{0.5\textwidth}
        \centering
        \includegraphics[width=\textwidth]{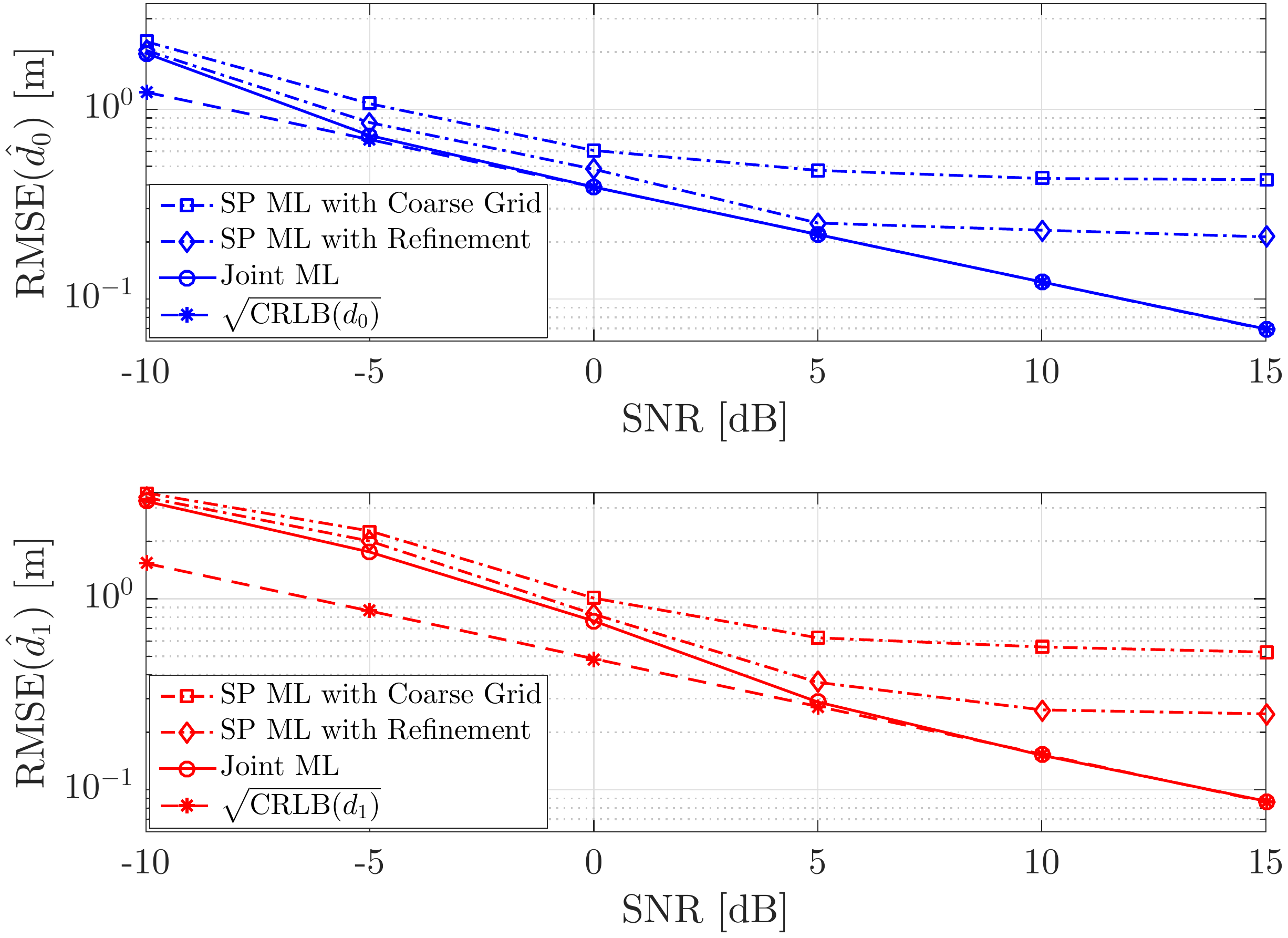}
        \caption{RMSE on the estimation of $d_k$.\label{fig::LMR5_dist}}
    \end{subfigure}%
    ~ 
    \begin{subfigure}[t]{0.5\textwidth}
        \centering
        \includegraphics[width=\textwidth]{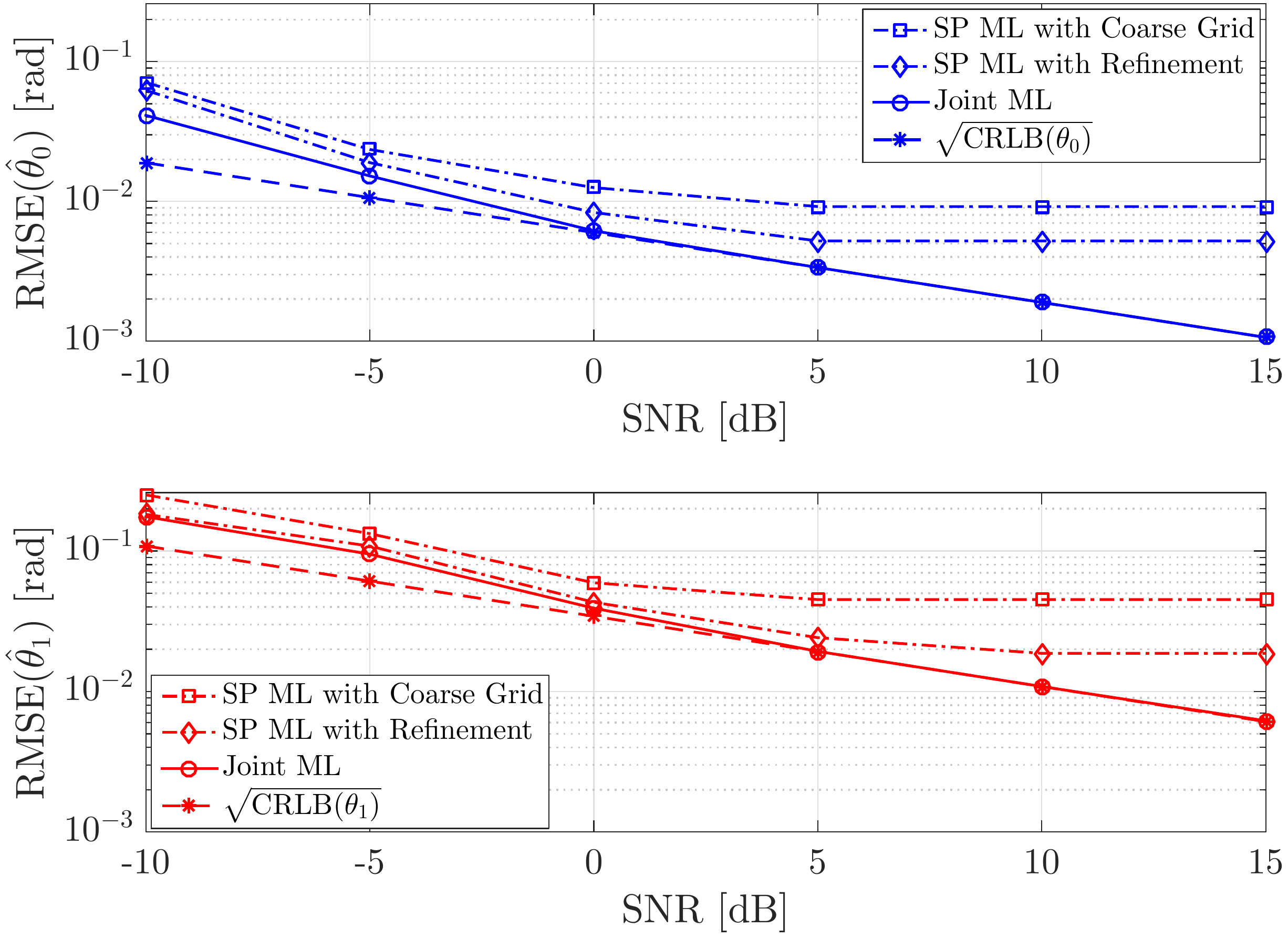}
        \caption{RMSE on the estimation of $\theta_k$.\label{fig::LMR5_AOD}}
    \end{subfigure}
    \caption{RMSEs on $d_k = c\tau_k$ and $\theta_k$ estimation in comparison with the CRLBs as a function of the SNR, for a $\text{LMR} = 5$ dB. \label{fig::LMR5}}
\end{figure}
 
\subsubsection{Performance assessment for LOS stronger than NLOS}we start the performance assessment by considering the more typical case in which the LOS path is received with a power greater than the NLOS, that is, we assume the presence of a single scatterer  at $\bm{s}_1 = [8 \ 13]^\text{T}$ m which produces a $\text{LMR} = 5$ dB. Fig.~\ref{fig::LMR5} reports the RMSEs on the estimation of the channel parameters $d_k$ and $\theta_k$, $k=0,1$, as a function of the SNR. The proposed estimator is labeled as 
``Joint ML'' and it is implemented in two-steps: in the first one, %the SAGE algorithm is applied to obtain 
an initial estimate of $\bm{\Theta}$ is obtained by searching for the $K+1$ dominant minima in $L_0(\theta,\tau)$ over a coarse $8 \times 8$ grid built from pairs $(\theta_k,\tau_k)$; the estimated vector $\hat{\bm{\Theta}}$ is then used to initialize a Nelder-Mead procedure which iteratively solves the $(K+1)$-dimensional ML estimation problem in \eqref{eq::finalJML}. For the sake of comparison, we also report the performance of the algorithms that approach the estimation problem by assuming a simplified single-path (SP) model. More precisely, we label as 
``SP ML with Coarse Grid'' the algorithm that simply %applies the SAGE method to 
optimizes $L_0(\theta,\tau)$ over the coarse $8 \times 8$ grid to estimate $\bm{\Theta}$ (i.e., the first step of the proposed Joint ML approach). The SP estimation performance can be further improved by using each estimated  $(\hat{\theta}_k,\hat{\tau}_k)$ pair in $\hat{\bm{\Theta}}$ to initialize a Nelder-Mead procedure that numerically optimize the SP cost function $L_0(\theta,\tau)$, yielding a refined estimate of $\bm{\Theta}$; in the following, we label such an approach as 
``SP ML with Refinement''. As concerns the theoretical lower bounds, each $\sqrt{\mathrm{CRLB}(\cdot)}$ is obtained  by inverting the FIM in either channel (ref. eq.(\ref{eq::FIMchannel})) or location (ref. eq.~(\ref{eq::FIMpos})) domain, selecting the corresponding diagonal entries and taking the square root.

By comparing the RMSEs in Fig.~\ref{fig::LMR5}, we observe that  the LOS channel parameters are estimated  more accurately than the  NLOS ones (as also reflected in the corresponding bounds), due to the stronger power of the former compared to latter. The "SP ML with Coarse Grid" algorithm (dash-dot curves with square markers) provides satisfactory initial estimates of both AODs and TOFs parameters, with an accuracy that increases with the SNR (since the powers of the LOS and NLOS paths increase accordingly) and with a reduced complexity thanks to the coarse grid used in the estimation process.
Although the performance further improves when a subsequent iterative 2D refinement is applied (see dash-dot curves with diamond markers), both the SP algorithms are still unable to achieve the theoretical lower bounds, as confirmed by the position errors reported in Fig.~\ref{fig:LMR5_position}. The existing gap clearly demonstrates that the algorithms derived assuming a simplified SP model cannot effectively cope with the residual mutual interference among the received paths. On the other hand, the solid curves show that the proposed Joint ML estimator offers the best performance: indeed, the RMSE of $\hat{\bm{p}}$ approaches the bound already for $\text{SNR} = -5$ dB, while the mapping of the scatterer position becomes increasingly more accurate until reaching the bound for $\text{SNR} = 5$ dB.

\begin{figure}
\centering
 \includegraphics[width=0.5\textwidth]{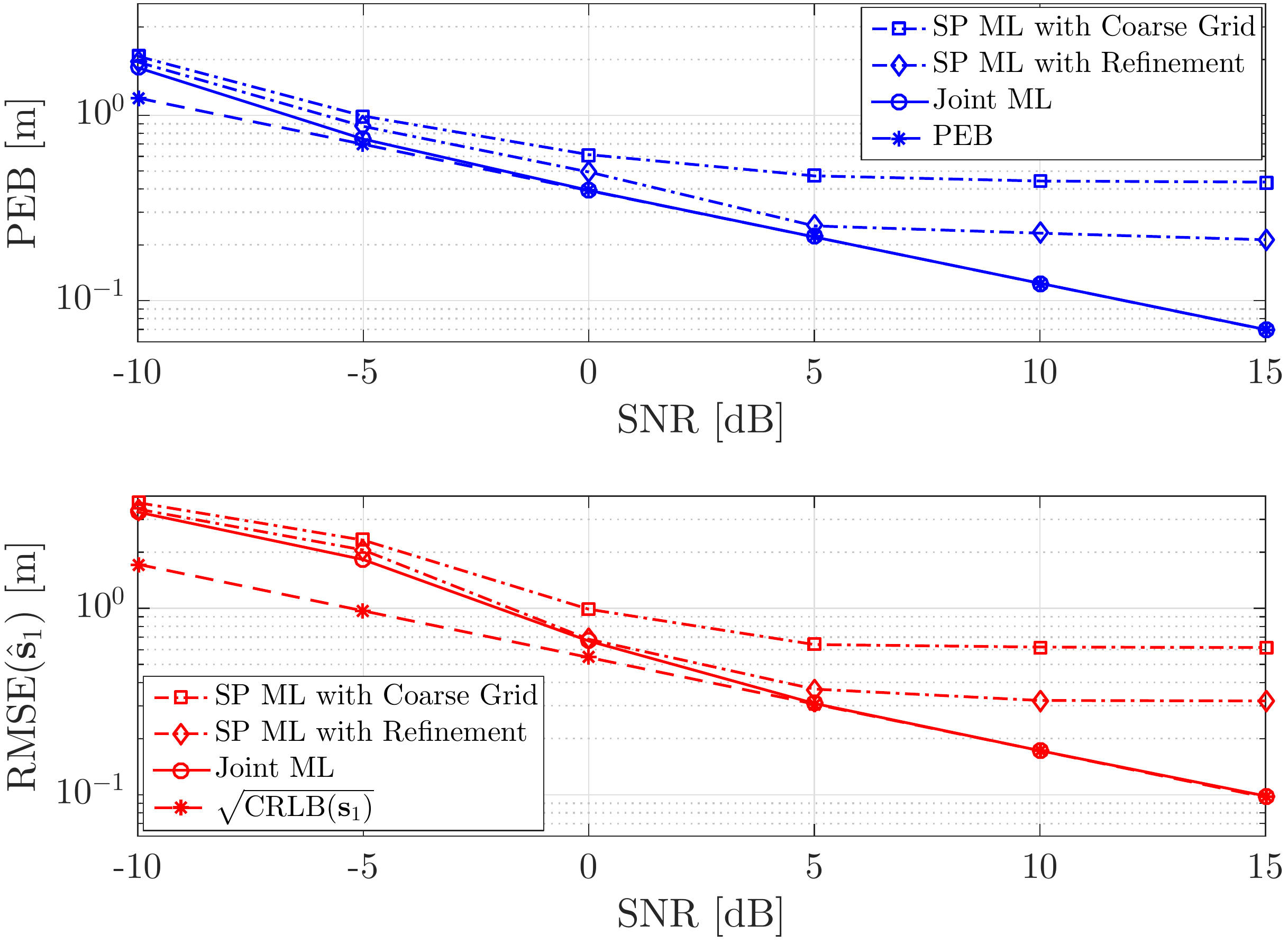}
 	\caption{RMSEs on MS and scatterer position estimation versus CRLBs as a function of the $\text{SNR}$, for $\text{LMR} = 5$ dB.}
\label{fig:LMR5_position}
 \end{figure}
 
\subsubsection{Performance assessment for LOS weaker than NLOS}to challenge the proposed Joint ML estimator, we consider the case in which the power of the NLOS path is $5$ dB higher than that of the LOS, that is, we set $\text{LMR} = -5$ dB. This setup is representative of scenarios in which the LOS path is severely attenuated. The RMSEs of $\hat{\bm{p}}$ and $\hat{\bm{s}}_1$ are reported in Fig.~\ref{fig:LMR-5_position}:  in this case, the higher power in the NLOS path translates into more advantageous conditions for mapping the scatterer position, which in fact  is more accurately estimated compared to the MS position, as confirmed by the smaller values of the bounds (dashed curves).
\begin{figure}
\centering
 \includegraphics[width=0.5\textwidth]{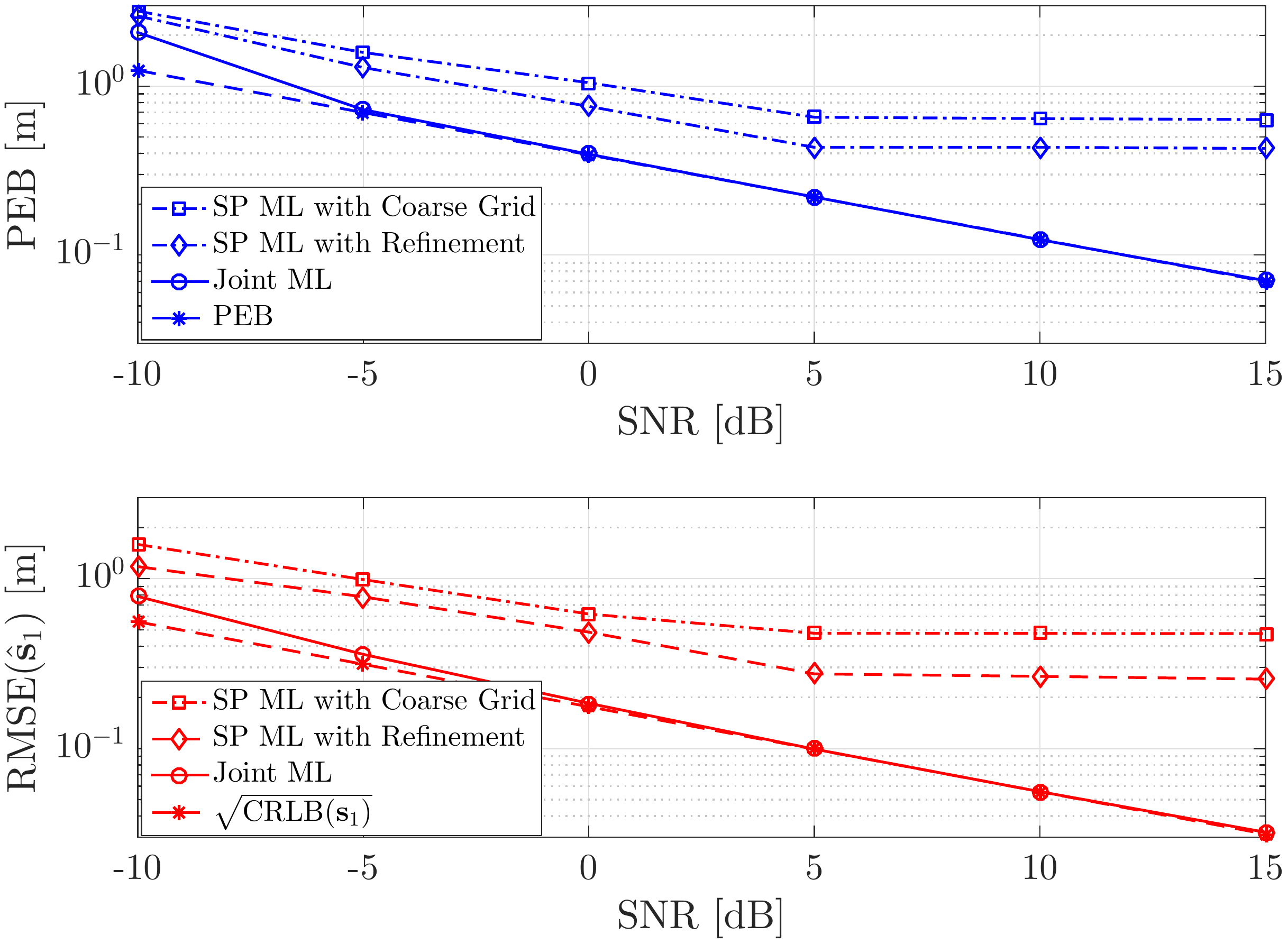}
 	\caption{RMSEs on MS and scatterer position estimation versus CRLBs as a function of the $\text{SNR}$, for $\text{LMR} = -5$ dB.}
\label{fig:LMR-5_position}
 \end{figure}
As it can be observed,  in this case the performances of the SP algorithms significantly deviate from the theoretical bounds. Remarkably, the proposed Joint ML estimator performs well even when the LOS path is highly attenuated, providing a very accurate  localization of the MS and mapping of the scatterer already at about 0 dB SNR.

%The fact that the accuracy does not improve as the SNR increases could be linked to the fact that both estimators are designed assuming a simplified (LOS-only) model, hence they do not face the estimation problem by considering all the unknown parameters in a joint fashion.  

\subsubsection{Performance assessment as a function of the multipath power}to corroborate the above results, we further analyze the algorithms behavior assuming a fixed value of the SNR and varying the multipath power in terms of LMR between $-10$ dB and $10$ dB, so as to obtain performance representative of a number of different operational conditions. In Fig.~\ref{fig::SNR10_position}, we show the RMSEs on the estimation of $\bm{p}$ and $\bm{s}_1$ as a function of the LMR, for a $\text{SNR} = 10$ dB. In agreement with the theoretical findings in Sec.~\ref{sec::roleNLOS} as well as with the analysis reported in Fig.~\ref{fig:PEB}, the PEB remains practically constant as the LMR changes, thus confirming the very weak dependency of LOS on the NLOS path, that is, the estimation of $\bm{p}$ is not harmed by the presence of multipath propagation. On the other hand, the dashed (red) curve shows that the accuracy achievable in the estimation of $\bm{s}_1$ progressively worsens as the power of the multipath diminishes. As it can be noticed, the performances of the SP algorithms are in trade-off: indeed, the RMSEs on the estimation of $\bm{p}$ tend to decrease as the LMR increases; conversely, the RMSEs of $\hat{\bm{s}}_1$ experiences an evident increase as the power of the NLOS path drops. Again, this behavior confirms that the performances are better for the more powerful path. Interestingly, the proposed Joint ML approach, thanks to its optimality, is able to cope with the less accurate scatterer position estimate for high LMRs, and vice versa with the less accurate MS position estimate for low LMRs. Indeed, the solid curves show that the joint estimator enables a satisfactory instantaneous localization and mapping in all the different operating conditions, significantly outperforming the SP competitors and attaining the bounds for even moderate values of the SNR.

%\begin{figure}
  %  \centering
   % \begin{subfigure}[t]{0.4\textwidth}
    %    \centering
     %   \includegraphics[width=\textwidth]{est_distances.pdf}
      %  \caption{RMSE on the estimation of $d_k = c\tau_k$.}
    %\end{subfigure}
    %\begin{subfigure}[t]{0.4\textwidth}
     %   \centering
      %  \includegraphics[width=\textwidth]{est_AODs.pdf}
       % \caption{RMSE on the estimation of $\theta_k$.}
%    \end{subfigure}
 %   \caption{RMSE on distance and AOD estimation for both matched-filter and ML 2D algorithms as a function of the LMR for SNR = 15 dB, also in comparison with the theoretical lower bounds.\label{fig::MatchVsML_LMR}}
%\end{figure}
\begin{figure}
\centering
 \includegraphics[width=0.5\textwidth]{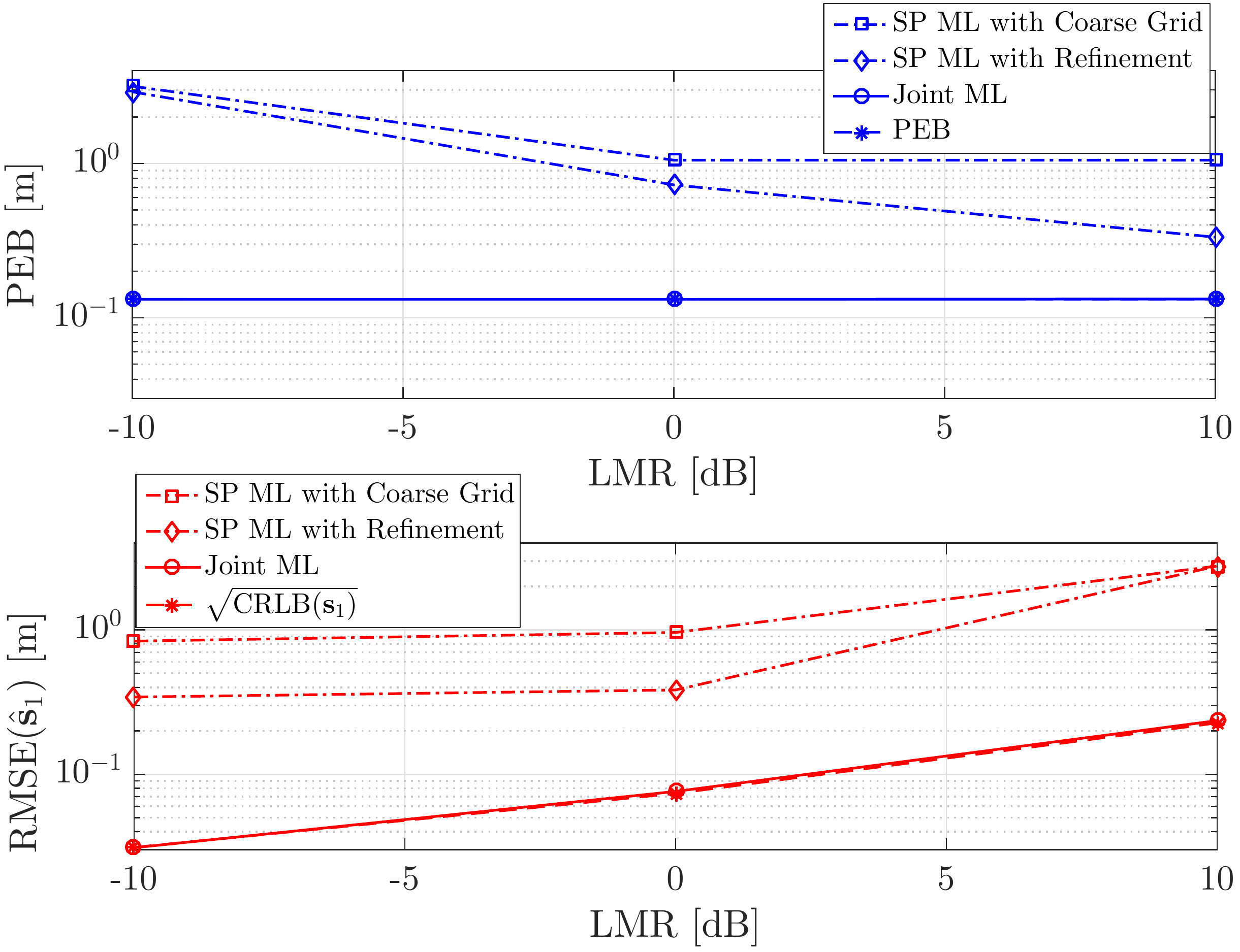}
 	\caption{RMSEs on MS and scatterer position estimation versus CRLBs as a function of the LMR, for $\text{SNR} = 10$ dB.}
\label{fig::SNR10_position}
 \end{figure}

\section{Conclusion}
The problem of single-snapshot estimation of the unknown MS position and mapping of scatterers locations in a mmWave MISO system has been addressed. The localization  process is based  on the combined use of AOD and TOF information, which can be estimated from a single pilot signal broadcast in DL by a BS. The Fisher information analysis demonstrated that localization and mapping is still possible also when using a single-antenna receiver but, differently from the MIMO setup, NLOS information cannot be used to improve the estimation of the MS position. We formulated the joint ML estimation problem in the channel domain and proposed and evaluated a low-complexity initialization method, which has an equivalent formulation in the position domain.  
%, derived its solution and showed that such formulation is equivalent to the one obtained directly in the position domain. To circumvent the need of a full-dimensional optimization, we proposed a novel two-step approach that exploits the sparsity of mmWave channels and is able to attain the theoretical bounds even when the LOS path is severely attenuated and the SNR is low.
% or


\appendices

\allowdisplaybreaks
\section{Derivation of  FIM elements in \eqref{eq::FIMblock} }\label{appA}
In the following, we provide the exact expressions of the entries of the FIM matrix in \eqref{eq::FIMblock}, derived based on \eqref{eq::FIMelemformula}. 
We introduce $\kappa_n=2\pi n /(NT_S)$, $\beta_{h,\ell}=\frac{2N_{\text{\tiny BS}}}{\sigma^2}\alpha^*_h \alpha_\ell \exp({{j \kappa_n (\tau_h - \tau_\ell)}})$ and $\bm{A}_{h,\ell}=\bm{a}_{\text{\tiny BS}}(\theta_h)\bm{a}^\text{H}_{\text{\tiny BS}}(\theta_\ell)$.
We start from the elements linked to the position-related  parameters $\theta_k$ and $\tau_k$, which are given by
\begin{align}
\Lambda(\tau_h,\tau_\ell) & =  \sum_{g,n}\Re\left\{\beta_{h,\ell} \kappa^2_n
(\bm{z}^g[n])^\text{H}\bm{A}_{h,\ell}\bm{z}^g[n] \right\}, \nonumber\\
\Lambda(\theta_h,\theta_\ell) & =  \sum_{g,n}\Re\left\{\beta_{h,\ell}(\bm{z}^g[n])^\text{H} 
\bm{D}^\text{H}_h\bm{A}_{h,\ell}\bm{D}_\ell\bm{z}^g[n] 
\right\}, \nonumber\\
\Lambda(\tau_h,\theta_\ell) & =  \sum_{g,n}\Re\left\{j  \kappa_n \beta_{h,\ell} (\bm{z}^g[n])^H \bm{A}_{h,\ell}\bm{D}_\ell\bm{z}^g[n] \right\}, \nonumber
\end{align}
%
%\begin{multline}
%\Lambda(\tau_h,\tau_\ell) = \frac{2N_{\text{\tiny BS}}}{\sigma^2} \sum_{g,n}\Re\left\{\beta_{h,\ell} \kappa^2_n
%\right. \\ \left. \times 
%(\bm{z}^g[n])^\text{H}\bm{A}_{h,\ell}\bm{z}^g[n] \right\}, \nonumber
%\end{multline}
%\begin{multline}
%\Lambda(\theta_h,\theta_\ell) = \frac{2N_{\text{\tiny BS}}}{\sigma^2} \sum_{g,n}\Re\left\{\beta_{h,\ell}(\bm{z}^g[n])^\text{H} 
%\right. \\ \left. \times 
%\bm{D}^\text{H}_h\bm{A}_{h,\ell}\bm{D}_\ell\bm{z}^g[n] 
%\right\}, \nonumber
%\end{multline}
%\begin{multline}
%\Lambda(\tau_h,\theta_\ell) = \frac{2N_{\text{\tiny BS}}}{\sigma^2} \sum_{g=1}^G %\sum_{n=0}^{N-1}\Re\left\{j  2\pi n/(NT_S)\beta_{h,\ell}
%\right. \\ \left. \times (\bm{z}^g[n])^H \bm{A}_{h,\ell}\bm{D}_\ell\bm{z}^g[n] \right\},
%\end{multline}
where the matrix $\bm{D}_u$ with subscript $u$ replaced by either $h$ or $\ell$ is given by
$$
\bm{D}_u = -j\frac{2\pi}{\lambda_c}d \cos \theta_u \mathrm{diag}[0\ 1\ \cdots\ (N_{\text{\tiny BS}}-1)]
$$
with $\mathrm{diag}(\point)$ a function which constructs a diagonal matrix with its entries. 

The elements including the channel amplitudes $r_k$ and phases $\phi_k$ are obtained as
\begin{align}
\Lambda(\tau_h,r_\ell) & =  \sum_{g,n}\Re\{j\e^{j\phi_\ell} \frac{\beta_{h,\ell}}{\alpha_{\ell}} \kappa_n  (\bm{z}^g[n])^H \bm{A}_{h,\ell}\bm{z}^g[n]\}, \nonumber\\
\Lambda(\tau_h,\phi_\ell) & =  \sum_{g,n}\Re\{-\beta_{h,\ell} \kappa_n  (\bm{z}^g[n])^H \bm{A}_{h,\ell}\bm{z}^g[n] \},\nonumber \\
\Lambda(\theta_h,r_\ell) & =  \sum_{g,n}\Re\{j\e^{j\phi_\ell} \frac{\beta_{h,\ell}}{\alpha_{\ell}}(\bm{z}^g[n])^H  \bm{D}^H_h \bm{A}_{h,\ell}\bm{z}^g[n] \}, \nonumber \\
\Lambda(\theta_h,\phi_\ell) & =  \sum_{g,n}\Re\{j\beta_{h,\ell}(\bm{z}^g[n])^H  \bm{D}^H_h\bm{A}_{h,\ell}\bm{z}^g[n] \},\nonumber \\
\Lambda(r_h,r_\ell) &=    \sum_{g,n}\Re\{\frac{\beta_{h,\ell}}{\alpha^*_h\alpha_\ell}\e^{j(\phi_\ell- \phi_h)}(\bm{z}^g[n])^H\bm{A}_{h,\ell}\bm{z}^g[n] \}, \nonumber \\
\Lambda(\phi_h,\phi_\ell) &=  \sum_{g,n} \Re\{\beta_{h,\ell}  (\bm{z}^g[n])^H\bm{A}_{h,\ell}\bm{z}^g[n] \}, \nonumber \\
\Lambda(r_h,\phi_\ell) &=   \sum_{g,n} \Re\{j\frac{\beta_{h,\ell}}{\alpha^*_h}\e^{-j\phi_h} (\bm{z}^g[n])^H\bm{A}_{h,\ell}\bm{z}^g[n] \}. \nonumber
\end{align}

% \begin{multline}
% \Lambda(\theta_h,\phi_\ell) = \frac{2N_{\text{\tiny BS}}}{\sigma^2} \sum_{g=1}^G \sum_{n=0}^{N-1}\Re\left\{j\beta_{h,\ell}(\bm{z}^g[n])^H \right. \\ \left. \times  \bm{D}^H_h\bm{A}_{h,\ell}\bm{z}^g[n] \right\},
% \end{multline}
% \begin{multline}
% \Lambda(r_h,r_\ell) = \frac{2N_{\text{\tiny BS}}}{\sigma^2} \sum_{g=1}^G \sum_{n=0}^{N-1}\Re\left\{\e^{j\left((\phi_\ell- \phi_h) + \kappa_n (\tau_\ell - \tau_h)\right)}  \right. \\ \left. \times (\bm{z}^g[n])^H\bm{A}_{h,\ell}\bm{z}^g[n] \right\},
% \end{multline}
% \begin{multline}
%  \frac{2N_{\text{\tiny BS}}}{\sigma^2} \sum_{g=1}^G \sum_{n=0}^{N-1}\Re\left\{\beta_{h,\ell} \right. \\ \left. \times (\bm{z}^g[n])^H\bm{A}_{h,\ell}\bm{z}^g[n] \right\},
% \end{multline}
% \begin{multline}
% \Lambda(r_h,\phi_\ell) =  \frac{2N_{\text{\tiny BS}}}{\sigma^2} \sum_{g=1}^G \sum_{n=0}^{N-1}\Re\left\{j \alpha_\ell \e^{j \left(\kappa_n(\tau_h - \tau_\ell) -\phi_h\right)} \right. \\ \left.  \times (\bm{z}^g[n])^H\bm{A}_{h,\ell}\bm{z}^g[n] \right\}.
% \end{multline}



\ifCLASSOPTIONcaptionsoff
  \newpage
\fi



%\bibliographystyle{IEEEtran}
\balance
\bibliographystyle{IEEEbib}
\bibliography{refs}

\begin{thebibliography}{10}

\bibitem{rappaport_5G}
T.~S. {Rappaport}, S.~{Sun}, R.~{Mayzus}, H.~{Zhao}, Y.~{Azar}, K.~{Wang},
  G.~N. {Wong}, J.~K. {Schulz}, M.~{Samimi}, and F.~{Gutierrez},
\newblock ``Millimeter wave mobile communications for {5G} cellular: It will
  work!,''
\newblock {\em IEEE Access}, vol. 1, pp. 335--349, 2013.

\bibitem{Witrisal}
K.~{Witrisal}, P.~{Meissner}, E.~{Leitinger}, Y.~{Shen}, C.~{Gustafson},
  F.~{Tufvesson}, K.~{Haneda}, D.~{Dardari}, A.~F. {Molisch}, A.~{Conti}, and
  M.~Z. {Win},
\newblock ``High-accuracy localization for assisted living: {5G} systems will
  turn multipath channels from foe to friend,''
\newblock {\em IEEE Signal Processing Magazine}, vol. 33, no. 2, pp. 59--70,
  2016.

\bibitem{Swindle}
L.~{Lu}, G.~Y. {Li}, A.~L. {Swindlehurst}, A.~{Ashikhmin}, and R.~{Zhang},
\newblock ``An overview of massive {MIMO}: {B}enefits and {C}hallenges,''
\newblock {\em IEEE Journal of Selected Topics in Signal Processing}, vol. 8,
  no. 5, pp. 742--758, 2014.

\bibitem{Heath}
T.~{Bai} and R.~W. {Heath},
\newblock ``Coverage and rate analysis for millimeter-wave cellular networks,''
\newblock {\em IEEE Transactions on Wireless Communications}, vol. 14, no. 2,
  pp. 1100--1114, 2015.

\bibitem{Tufve}
E.~G. {Larsson}, O.~{Edfors}, F.~{Tufvesson}, and T.~L. {Marzetta},
\newblock ``Massive {MIMO} for next generation wireless systems,''
\newblock {\em IEEE Communications Magazine}, vol. 52, no. 2, pp. 186--195,
  2014.

\bibitem{5G_vehic}
H.~{Wymeersch}, G.~{Seco-Granados}, G.~{Destino}, D.~{Dardari}, and
  F.~{Tufvesson},
\newblock ``{5G} mm{W}ave positioning for vehicular networks,''
\newblock {\em IEEE Wireless Communications}, vol. 24, no. 6, pp. 80--86, 2017.

\bibitem{Peral}
J.~A. {del Peral-Rosado}, R.~{Raulefs}, J.~A. {Lopez-Salcedo}, and
  G.~{Seco-Granados},
\newblock ``Survey of cellular mobile radio localization methods: From {1G} to
  {5G},''
\newblock {\em IEEE Communications Surveys Tutorials}, vol. 20, no. 2, pp.
  1124--1148, 2018.

\bibitem{Kakkavas}
A.~{Kakkavas}, M.~H. {Casta\~{n}eda Garc\'{i}a}, R.~A. {Stirling-Gallacher},
  and J.~A. {Nossek},
\newblock ``Performance limits of single-anchor millimeter-wave positioning,''
\newblock {\em IEEE Transactions on Wireless Communications}, vol. 18, no. 11,
  pp. 5196--5210, 2019.

\bibitem{AbuShaban}
Z.~{Abu-Shaban}, X.~{Zhou}, T.~{Abhayapala}, G.~{Seco-Granados}, and
  H.~{Wymeersch},
\newblock ``Error bounds for uplink and downlink {3D} localization in {5G}
  millimeter wave systems,''
\newblock {\em IEEE Transactions on Wireless Communications}, vol. 17, no. 8,
  pp. 4939--4954, 2018.

\bibitem{Arash}
A.~{Shahmansoori}, G.~E. {Garcia}, G.~{Destino}, G.~{Seco-Granados}, and
  H.~{Wymeersch},
\newblock ``Position and orientation estimation through millimeter-wave {MIMO}
  in {5G} systems,''
\newblock {\em IEEE Transactions on Wireless Communications}, vol. 17, no. 3,
  pp. 1822--1835, 2018.

\bibitem{Nil}
N.~{Garcia}, H.~{Wymeersch}, E.~G. {Larsson}, A.~M. {Haimovich}, and
  M.~{Coulon},
\newblock ``Direct localization for massive {MIMO},''
\newblock {\em IEEE Transactions on Signal Processing}, vol. 65, no. 10, pp.
  2475--2487, 2017.

\bibitem{ESPRIT}
A.~{Hu}, T.~{Lv}, H.~{Gao}, Z.~{Zhang}, and S.~{Yang},
\newblock ``An {ESPRIT}-based approach for {2-D} localization of incoherently
  distributed sources in massive {MIMO} systems,''
\newblock {\em IEEE Journal of Selected Topics in Signal Processing}, vol. 8,
  no. 5, pp. 996--1011, 2014.

\bibitem{Koivisto}
M.~{Koivisto}, M.~{Costa}, J.~{Werner}, K.~{Heiska}, J.~{Talvitie},
  K.~{Leppänen}, V.~{Koivunen}, and M.~{Valkama},
\newblock ``Joint device positioning and clock synchronization in {5G}
  ultra-dense networks,''
\newblock {\em IEEE Transactions on Wireless Communications}, vol. 16, no. 5,
  pp. 2866--2881, 2017.

\bibitem{Palacios}
J.~{Palacios}, G.~{Bielsa}, P.~{Casari}, and J.~{Widmer},
\newblock ``Single- and multiple-access point indoor localization for
  millimeter-wave networks,''
\newblock {\em IEEE Transactions on Wireless Communications}, vol. 18, no. 3,
  pp. 1927--1942, 2019.

\bibitem{Mendrzik}
R.~{Mendrzik}, H.~{Wymeersch}, G.~{Bauch}, and Z.~{Abu-Shaban},
\newblock ``Harnessing {NLOS} components for position and orientation
  estimation in {5G} millimeter wave {MIMO},''
\newblock {\em IEEE Transactions on Wireless Communications}, vol. 18, no. 1,
  pp. 93--107, 2019.

\bibitem{SLAM1}
H.~{Durrant-Whyte} and T.~{Bailey},
\newblock ``Simultaneous localization and mapping: part {I},''
\newblock {\em IEEE Robotics \& Automation Magazine}, vol. 13, no. 2, pp.
  99--110, 2006.

\bibitem{SLAM2}
T.~{Bailey} and H.~{Durrant-Whyte},
\newblock ``Simultaneous localization and mapping ({SLAM}): part {II},''
\newblock {\em IEEE Robotics \& Automation Magazine}, vol. 13, no. 3, pp.
  108--117, 2006.

\bibitem{Destino}
J.~{Talvitie}, M.~{Valkama}, G.~{Destino}, and H.~{Wymeersch},
\newblock ``Novel algorithms for high-accuracy joint position and orientation
  estimation in {5G} mm{W}ave systems,''
\newblock in {\em 2017 IEEE Globecom Workshops (GC Wkshps)}, 2017, pp. 1--7.

\bibitem{Mendrzik2}
R.~{Mendrzik}, H.~{Wymeersch}, and G.~{Bauch},
\newblock ``Joint localization and mapping through millimeter wave {MIMO} in
  {5G} systems,''
\newblock in {\em 2018 IEEE Global Communications Conference (GLOBECOM)}, 2018,
  pp. 1--6.

\bibitem{Battistelli}
H.~{Kim}, K.~{Granstr\"{o}m}, L.~{Gao}, G.~{Battistelli}, S.~{Kim}, and
  H.~{Wymeersch},
\newblock ``{5G} mmwave cooperative positioning and mapping using multi-model
  {PHD} filter and map fusion,''
\newblock {\em IEEE Transactions on Wireless Communications}, vol. 19, no. 6,
  pp. 3782--3795, 2020.

\bibitem{Leitinger2}
X.~{Li}, E.~{Leitinger}, M.~{Oskarsson}, K.~{\r{A}str\"{o}m}, and
  F.~{Tufvesson},
\newblock ``Massive {MIMO}-based localization and mapping exploiting phase
  information of multipath components,''
\newblock {\em IEEE Transactions on Wireless Communications}, vol. 18, no. 9,
  pp. 4254--4267, 2019.

\bibitem{Westberg}
E.~{Westberg}, J.~{Staudinger}, J.~{Annes}, and V.~{Shilimkar},
\newblock ``{5G} infrastructure {RF} solutions: Challenges and opportunities,''
\newblock {\em IEEE Microwave Magazine}, vol. 20, no. 12, pp. 51--58, 2019.

\bibitem{Gentner}
C.~{Gentner}, T.~{Jost}, W.~{Wang}, S.~{Zhang}, A.~{Dammann}, and U.~{Fiebig},
\newblock ``Multipath assisted positioning with simultaneous localization and
  mapping,''
\newblock {\em IEEE Transactions on Wireless Communications}, vol. 15, no. 9,
  pp. 6104--6117, 2016.

\bibitem{Leitinger_SLAM}
E.~{Leitinger}, F.~{Meyer}, F.~{Hlawatsch}, K.~{Witrisal}, F.~{Tufvesson}, and
  M.~Z. {Win},
\newblock ``A belief propagation algorithm for multipath-based {SLAM},''
\newblock {\em IEEE Transactions on Wireless Communications}, vol. 18, no. 12,
  pp. 5613--5629, 2019.

\bibitem{ICASSP2020}
A.~{Fascista}, A.~{Coluccia}, H.~{Wymeersch}, and G.~{Seco-Granados},
\newblock ``Low-complexity accurate mmwave positioning for single-antenna users
  based on angle-of-departure and adaptive beamforming,''
\newblock in {\em IEEE International Conference on Acoustics, Speech and Signal
  Processing (ICASSP)}, 2020, pp. 4866--4870.

\bibitem{TWC_2019}
A.~{Fascista}, A.~{Coluccia}, H.~{Wymeersch}, and G.~{Seco-Granados},
\newblock ``Millimeter-wave downlink positioning with a single-antenna
  receiver,''
\newblock {\em IEEE Transactions on Wireless Communications}, vol. 18, no. 9,
  pp. 4479--4490, 2019.

\bibitem{Alkhateeb1}
A.~{Alkhateeb} and R.~W. {Heath},
\newblock ``Frequency selective hybrid precoding for limited feedback
  millimeter wave systems,''
\newblock {\em IEEE Transactions on Communications}, vol. 64, no. 5, pp.
  1801--1818, 2016.

\bibitem{brady1}
J.~{Brady}, N.~{Behdad}, and A.~M. {Sayeed},
\newblock ``Beamspace {MIMO} for millimeter-wave communications: System
  architecture, modeling, analysis, and measurements,''
\newblock {\em IEEE Transactions on Antennas and Propagation}, vol. 61, no. 7,
  pp. 3814--3827, 2013.

\bibitem{brady2}
J.~H. {Brady} and A.~M. {Sayeed},
\newblock ``Wideband communication with high-dimensional arrays: New results
  and transceiver architectures,''
\newblock in {\em IEEE International Conference on Communication Workshop
  (ICCW)}, 2015, pp. 1042--1047.

\bibitem{AbuSync}
Z.~{Abu-Shaban}, H.~{Wymeersch}, T.~{Abhayapala}, and G.~{Seco-Granados},
\newblock ``Single-anchor two-way localization bounds for {5G} mmwave
  systems,''
\newblock {\em IEEE Transactions on Vehicular Technology}, pp. 1--1, 2020.

\bibitem{Sark}
V.~{Sark}, E.~{Grass}, and J.~{Gutiérrez},
\newblock ``Multi-way ranging with clock offset compensation,''
\newblock in {\em Advances in Wireless and Optical Communications (RTUWO)},
  2015, pp. 68--71.

\bibitem{cox1987parameter}
D.~R. {Cox} and N.~{Reid},
\newblock ``Parameter orthogonality and approximate conditional inference,''
\newblock {\em Journal of the Royal Statistical Society: Series B
  (Methodological)}, vol. 49, no. 1, pp. 1--18, 1987.

\bibitem{Leitinger}
E.~{Leitinger}, P.~{Meissner}, C.~{Rüdisser}, G.~{Dumphart}, and
  K.~{Witrisal},
\newblock ``Evaluation of position-related information in multipath components
  for indoor positioning,''
\newblock {\em IEEE Journal on Selected Areas in Communications}, vol. 33, no.
  11, pp. 2313--2328, 2015.

\bibitem{Henk5G}
H.~{Wymeersch}, N.~{Garcia}, H.~{Kim}, G.~{Seco-Granados}, S.~{Kim}, F.~{Wen},
  and M.~{Fr\"ohle},
\newblock ``{5G} mmwave downlink vehicular positioning,''
\newblock in {\em 2018 IEEE Global Communications Conference (GLOBECOM)}, 2018,
  pp. 206--212.

\bibitem{Barati}
C.~N. {Barati}, S.~A. {Hosseini}, M.~{Mezzavilla}, T.~{Korakis}, S.~S.
  {Panwar}, S.~{Rangan}, and M.~{Zorzi},
\newblock ``Initial access in millimeter wave cellular systems,''
\newblock {\em IEEE Transactions on Wireless Communications}, vol. 15, no. 12,
  pp. 7926--7940, 2016.

\bibitem{Giordani}
M.~{Giordani}, M.~{Mezzavilla}, and M.~{Zorzi},
\newblock ``Initial access in {5G} mm{W}ave cellular networks,''
\newblock {\em IEEE Communications Magazine}, vol. 54, no. 11, pp. 40--47,
  2016.

\bibitem{NumReci}
William~H. Press, Saul~A. Teukolsky, William~T. Vetterling, and Brian~P.
  Flannery,
\newblock {\em Numerical Recipes 3rd Edition: The Art of Scientific Computing},
\newblock Cambridge University Press, USA, 3 edition, 2007.

\bibitem{Nelder}
Jeffrey~C. Lagarias, James~A. Reeds, Margaret~H. Wright, and Paul~E. Wright,
\newblock ``Convergence properties of the nelder--mead simplex method in low
  dimensions,''
\newblock {\em SIAM J. on Optimization}, vol. 9, no. 1, pp. 112--147, May 1998.

\bibitem{Li1}
Q.~C. {Li}, G.~{Wu}, and T.~S. {Rappaport},
\newblock ``Channel model for millimeter-wave communications based on geometry
  statistics,''
\newblock in {\em IEEE Globecom Workshops (GC Wkshps)}, 2014, pp. 427--432.

\bibitem{Li2}
Q.~{Li}, H.~{Shirani-Mehr}, T.~{Balercia}, A.~{Papathanassiou}, G.~{Wu},
  S.~{Sun}, M.~K. {Samimi}, and T.~S. {Rappaport},
\newblock ``Validation of a geometry-based statistical mmwave channel model
  using ray-tracing simulation,''
\newblock in {\em IEEE Vehicular Technology Conference (VTC Spring)}, 2015, pp.
  1--5.

\bibitem{Rappaport}
T.~S. {Rappaport}, Y.~{Xing}, G.~R. {MacCartney}, A.~F. {Molisch},
  E.~{Mellios}, and J.~{Zhang},
\newblock ``Overview of millimeter wave communications for fifth-generation
  {(5G)} wireless networks—with a focus on propagation models,''
\newblock {\em IEEE Transactions on Antennas and Propagation}, vol. 65, no. 12,
  pp. 6213--6230, 2017.

\bibitem{Sage}
J.~A. {Fessler} and A.~O. {Hero},
\newblock ``Space-alternating generalized expectation-maximization algorithm,''
\newblock {\em IEEE Transactions on Signal Processing}, vol. 42, no. 10, pp.
  2664--2677, 1994.

\bibitem{SageApp1}
B.~H. {Fleury}, M.~{Tschudin}, R.~{Heddergott}, D.~{Dahlhaus}, and K.~{Ingeman
  Pedersen},
\newblock ``Channel parameter estimation in mobile radio environments using the
  {SAGE} algorithm,''
\newblock {\em IEEE Journal on Selected Areas in Communications}, vol. 17, no.
  3, pp. 434--450, 1999.

\bibitem{SageApp2}
A.~{Pin}, R.~{Rinaldo}, A.~{Tonello}, C.~{Marshall}, M.~{Driusso}, A.~{Biason},
  and A.~D. {Torre},
\newblock ``{LTE} ranging measurement using uplink opportunistic signals and
  the {SAGE} algorithm,''
\newblock in {\em 27th European Signal Processing Conference (EUSIPCO)}, 2019,
  pp. 1--5.

\end{thebibliography}
\end{document}